\documentclass[12pt]{article}
\usepackage{epsfig,multicol,bbm,amsmath,amssymb,amscd,mathrsfs,fancybox,amsthm,framed,enumerate,txfonts, booktabs}
\usepackage{fancyhdr}

\usepackage[top=20truemm,bottom=25truemm,left=17truemm,right=17truemm]{geometry}
\usepackage{amsfonts}
\usepackage[mathscr]{eucal}
\usepackage{eucal}
  \makeatletter
    
    \@addtoreset{equation}{section}
  \makeatother

\newtheorem{Def}{Definition}[section]
\newtheorem{Thm}{Theorem}[section]
\newtheorem{Lem}{Lemma}[section]
\newtheorem{Prop}{Proposition}[section]
\newtheorem{Ass}{Assumption}[section]

\newcommand{\beq}{\begin{align}}
\newcommand{\eeq}{\end{align}}

\newcommand{\Expect}[1]{\left\langle{#1}\right\rangle}
\newcommand{\bvec}[1]{\boldsymbol{#1}}
\newcommand{\no}{\nonumber}

\newcommand{\Natural}{\mathbb{N}}

\newcommand{\Real}{\mathbb{R}}

\newcommand{\pdfrac}[2]{\frac{\partial #1}{\partial #2}}

\newcommand{\q}{\quad}
\renewcommand{\H}{\mathcal{H}}

\newcommand{\F}{\mathcal{F}}
\DeclareMathOperator*{\op}{\oplus}
\DeclareMathOperator*{\wg}{\wedge}
\DeclareMathOperator*{\ot}{\otimes}
\DeclareMathOperator*{\hop}{\hat{\oplus}}
\DeclareMathOperator*{\hot}{\hat{\otimes}}
\newcommand{\R}{\mathbb{R}}
\newcommand{\C}{\mathbb{C}}
\newcommand{\kk}{\mathbf{k}}
\newcommand{\pp}{\mathbf{p}}
\newcommand{\xx}{\mathbf{x}}
\newcommand{\yy}{\mathbf{y}}

\newcommand{\bb}{\mathrm{b}}
\newcommand{\ff}{\mathrm{f}}

\renewcommand{\S}{\mathcal{S}}
\renewcommand{\rm}{\mathrm}
\newcommand{\fr}{\frac}
\renewcommand{\hat}{\widehat}

\begin{document}
\title{Construction of dynamics and time-ordered exponential for unbounded non-symmetric Hamiltonians}

\maketitle
\begin{center}
Shinichiro Futakuchi and Kouta Usui
\vskip 2mm
\begin{small}
\textit{Department of Mathematics, Hokkaido University}\\
 \textit{060-0810, Sapporo, Japan.}

\end{small}
\end{center}
\abstract{We prove under certain assumptions that there exists a solution of
the Schr\"odinger or
the Heisenberg equation of motion generated by a linear operator 
$H$ acting in some complex
Hilbert space $\mathcal{H}$, which may be unbounded, \textit{not symmetric}, or \textit{not normal}.
We also prove that, under the same assumptions, 
 there exists a time evolution operator in the interaction picture and that the evolution operator
 enjoys a useful series expansion formula. This expansion is considered to be one of the mathematically rigorous 
 realizations of so called ``time-ordered exponential", which is familiar in the physics literature.  
We apply the general theory to prove the existence of dynamics for the mathematical model of Quantum Electrodynamics (QED)
quantized in the Lorenz gauge, the interaction Hamiltonian of which is not even symmetric or normal.

\section{Introduction}

Let $\mathcal{H}$ be a complex Hilbert space and $ H $ be a linear operator on $ \H $. We consider the initial value problem for the Schr\"odinger equation
\begin{align}\label{sch00}
\pdfrac{\xi(t)} t =-iH\xi(t),\quad \xi (0)=\xi , 
\end{align}
or for the Heisenberg equation 
\begin{align}\label{hei00}
\frac{dB(t)}{dt}=[iH, B(t)] ,\quad B(0)=B,
\end{align}
where $ B $ is a possibly unbounded linear operator on $ \H $, and $ [X,Y] := XY-YX $. In the context of quantum mechanics, 
the parameter $t\in\R$ represents time, and $H$ is regarded as a Hamiltonian of the quantum system under consideration.
At time $t\in\R$, $ \xi (t) $ or $ B(t) $ describes a time developed state vector or 
a time developed observable, respectively. Then, the general mathematical study of 
the initial value problems \eqref{sch00} or \eqref{hei00} is of great interest since it will reveal 
dynamics of a certain class of quantum systems.

In the ordinary formulation of quantum mechanics, a Hamiltonian $ H $ is assumed to be a self-adjoint operator. In this case, the solutions of these equations are given by 
\begin{align}
\xi (t) &=e^{-itH}\xi, \\
B(t)&=e^{itH}Be^{-itH},
\end{align}
with some suitable conditions for operator domains in the Heisenberg case (See \cite{MR2362899}, in detail). However, 
in some models, the Hamiltonian $ H $ may not be self-adjoint or not even normal. When $H$ is unbounded and not normal, the 
above time evolution operator $e^{-itH}$ does not immediately make sense since for unbounded $H$ it is usually defined through operational calculus. In such cases, it is not obvious at all that there exist solutions of these equations. 

The most important realistic examples 
that can cause this difficulty contain the mathematical model of Quantum Electrodymamics (QED) when it is quantized in a Lorentz covariant gauge such as Lorenz gauge \cite{MR2412280, MR2533876}. In the Lorenz-gauge QED, we have to adopt a vector space with an indefinite metric as a vector space of quantum mechanical state vectors, in order to realize the canonical commutation relations. Indefinite metric results in a non-symmetric Hamiltonian which is not even normal, and thus it is far from trivial that dynamics of the Lorenz-gauge QED really exists. To obtain dynamics for such models, one may apply the general theory of evolution equations or Cauchy problems by estimating the resolvent operators \cite{MR2108957,MR1721989}, but we will take another way to avoid hard resolvent estimates. The first motivation of the present study is to establish a general theory as to the existence of dynamics with Hamiltonians which is not symmetric and not even normal.

Another motivation of the present work also comes from quantum theory. We consider a system with a Hamiltonian of the type
\begin{align}\label{def-of-H}
H = H_0 +H_1,
\end{align}
where $ H_0 $ is a \textit{solvable} Hamiltonian (of which we already know the dynamics) and $ H_1 $ is an interaction Hamiltonian 
which causes unknown dynamics.
To study a quantum mechanical scattering problem with Hamiltonians of this form, 
it is often useful to employ the so called interaction picture, 
in which both state vectors and observables evolve in time. 
The evolution operator in the interaction picture from time $t'$ to time $t$ --- which is usually denoted by $U(t,t')$ --- is a solution of the differential equations
\begin{align}
\frac{\partial }{\partial t} U(t,t')  &= -i H_1 (t) U(t,t') ,\label{int-eqs} \\
\frac{\partial }{\partial t'} U(t,t')  &= i U(t,t') H_1 (t') , \label{int-eqs2}
\end{align}
with 
\begin{align}
H_1(t):=e^{itH_0}H_1e^{-itH_0}. 
\end{align}
It is easy to \textit{heuristically} derive the series expansion of $U(t,t')$
\begin{align}\label{formal-series-of-U}
U(t,t')=1+(-i)\int_{t'}^t d\tau_1\,H_1(\tau_1) +(-i)^2\int_{t'}^t  d\tau_1\int_{t'}^{\tau_1}d\tau_2\,H_1(\tau_1)H_1(\tau_2)+\dots.
\end{align}    
This expansion formula \eqref{formal-series-of-U} is well known to be quite useful in computing scattering amplitudes
of elementary particles such as electrons or photons, and the results dramatically agree with
the high energy experiments,
even though these computations contain a lot of mathematically unrigorous steps \cite{MR2148466,MR2148467,MR1402248}.  

The series expansion \eqref{formal-series-of-U} has already been rigorously analyzed,
in the case where $H_1$ is bounded (See, e.g., Refs.  \cite{MR552941}, \cite{MR0493420} Section X.12, \cite{MR0492656}, \cite{MR583242}, \cite{MR683026}).
However, in the case where $H_1$ is not bounded, it seems that there have been few mathematically rigorous studies of the series expansion \eqref{formal-series-of-U} in an abstract or a
general form.
The second motivation of the present work is to 
prove in mathematically rigorous manner with certain assumptions that there exists a time evolution operator $U(t,t')$ satisfying \eqref{int-eqs} and \eqref{int-eqs2} which possesses the series representation \eqref{formal-series-of-U} on certain dense subspace, including the case where $H_1$ is neither bounded nor normal. The solutions of Schr\"odinger or Heisenberg equation will be constructed by using the operator $U(t,t')$. Here, we stress that our proof does not only state the existence of the solutions abstractly but also derive an explicit series expansion \eqref{formal-series-of-U}
for them, which would be useful for the practical applications in mathematical physics.
 
This paper is organized as follows. 
In Section \ref{main-results}, we will summarize our results.
In Section \ref{construction}, a solution $U(t,t')$ of the differential equations \eqref{int-eqs}, \eqref{int-eqs2} is explicitly constructed. 
In Section \ref{properties}, we will derive several properties of the solution $U(t,t')$. 
In Section \ref{solutions}, we will construct solutions of Schr\"odinger and Heisenberg equations of motion. 
In Section \ref{QED}, QED quantized in the Lorenz gauge will be discussed and it will be proved that there exists a solution of the Heisenberg equations of motion for quantized fields.
In Appendix, some further mathematical properties of $U(t,t')$ will be studied.

\section{Main Results}\label{main-results}

The inner product and the norm of $ \H $ are denoted by $ \Expect{ \cdot , \cdot } _\H $ (anti-linear in the first variable) and $ \| \cdot \| _\H $ respectively. When there can be no danger of confusion, then the subscript $ \H $ in $ \Expect{ \cdot , \cdot } _\H $ and $ \| \cdot \| _\H $ is omitted. For a linear operator $ T $ in $ \H $, we denote its domain (resp. range) by $ D(T) $ (resp. $ R(T) $). We also denote the adjoint of $T$ by $ T^* $ and the closure by $ \bar{T} $ if these exist. For a self-adjoint operator $ T $,  $ E_T (\cdot ) $ denotes the spectral measure of $ T $.

Let $ H_0 $ be a self-adjoint operator on $ \H $ and $ H_1 $ be a densely defined closed operator on $ \H $. Set
\begin{align}\label{H}
H := H_0 + H_1 ,
\end{align}
with the domain $ D(H_0) \cap D(H_1) $.

First, we assume that there exists an operator $A$ in $\mathcal{H}$
satisfying the following conditions:
\begin{Ass}\label{ass1}
\begin{enumerate}[(I)]
\item $A$ is self-adjoint and non-negative.
\item $A$ and $H_0$ are strongly commuting.
\item $ H_1 $ is $A^{1/2} $- bounded, where $ A^{1/2} $ defined through operational calculus.
\item There exists a constant $b>0$ such that, for all $L\ge 0$, $\xi \in R( E_A([0,L]))$
implies $H_1 \xi \in R( E_A([0,L+b]))$.
\end{enumerate}
\end{Ass}
We remark that 
the above condition (IV) comes from the following
physical consideration.
Suppose that $H$ is the Hamiltonian of a certain quantum system.
The above self-adjoint operator $A$ is expected to be an observable quantity
of the quantum system under consideration, 
typically a particle number in application to quantum field theories (see 
application in Section \ref{QED}.).
Roughly speaking, the condition (IV) says that the value of the observable $A$
increase at most $b$ by one interaction.

Hereafter, we use the following notations:
\begin{align}
& V_L := R( E_A ([0, L])) , \q L \ge 0, \\
& D := \bigcup _{L\ge 0} V_L . 
\end{align}
Since $A$ is assumed to be self-adjoint, it follows that $D$ is a dense subspace in $\mathcal{H}$.

Our first result is:
\begin{Thm}\label{main-thm1} Under Assumption \ref{ass1}, for each $ t,t' \in \R , \, \xi \in D $, the series:
\begin{align}
U(t,t') \xi := \xi + (-i) \int _{t'} ^t d\tau _1 \, H_1(\tau _1) \xi + (-i)^2 \int _{t'} ^t d\tau _1 \int _{t'} ^{\tau _1} d\tau _2 \, H_1 (\tau _1) H_1 (\tau _2) \xi + \cdots 
\end{align}
converges absolutely, where each of integrals are strong integrals. Furthermore, $ U(t,t') $ forms an evolution operator on $ D $, that is, the following (i), (ii) hold.
\begin{enumerate}[(i)]
\item  For fixed $ t' \in \Real  $ and $ \xi \in D $, the vector valued function $ \R\ni t\mapsto U(t,t') \xi $ is strongly continuously differentiable, 
and $ U(t,t') \xi \in D(H_1 (t)) $. Moreover, $ U(t,t') \xi $ satisfies
\begin{align}\label{DE1}
\frac{\partial }{\partial t} U(t,t') \xi = -i H_1 (t) U(t,t') \xi . 
\end{align}
\item  For fixed $ t \in \Real  $ and $ \xi \in D $, the vector valued function $\R \ni t'\mapsto  U(t,t') \xi $ is strongly continuously differentiable,
 and satisfies
\begin{align}\label{DE2}
\frac{\partial }{\partial t'} U(t,t') \xi = i U(t,t') H_1 (t') \xi .
\end{align}
\end{enumerate}
 \end{Thm}

Next, we assume the following properties in addition.
\begin{Ass}\label{ass2}
\begin{enumerate}[(I)]
\item $H_1^* $ is $A^{1/2}$- bounded.
\item There exists a constant $b'>0$ such that, for all $L\ge 0$, $\xi \in R( E_A([0,L]))$
implies $H_1 ^* \xi \in R( E_A([0,L+b'])) $.
\end{enumerate}
\end{Ass}

Then, we have
\begin{Thm}\label{adjoint} Under Assumptions \ref{ass1}-\ref{ass2}, it follows that $ D\subset D(U(t,t') ^*) $, and for all $\xi\in D$, $U(t,t')^*\xi$ is strongly continuously differentiable with respect to $t$ and $t'$, and satisfies
\begin{align}
\label{adjoint4} & \frac{\partial }{\partial t} U(t,t') ^* \xi = i U(t,t')^* H_1 (t) ^* \xi , \\
\label{adjoint5} & \frac{\partial }{\partial t'} U(t,t') ^* \xi = -i H_1 (t') ^* U(t,t') ^* \xi .
\end{align}
In particular, $U(t,t')$ is closable.
\end{Thm}

The time evolution operator $U(t,t')$ has the following properties.
\begin{Thm}\label{main-thm2} Under Assumptions \ref{ass1}-\ref{ass2}, the following (i) and (ii) hold.
\begin{enumerate}[(i)]
\item For all $ \xi \in D , \, t,t',t'' \in \R $, $ U(t,t) \xi = \xi $ and the operator equality
\begin{align}\label{associative}
\overline{U(t,t')} U(t' , t'') = U(t,t'').
\end{align}
holds.
\item For any $ s, t,t' \in \Real $, the operator equality
\begin{align}
e^{isH_0} \overline{U(t,t')} e^{-isH_0} = \overline{U(t+s, t'+s)}
\end{align}
holds.
\end{enumerate}
\end{Thm}

If we assume in addition that $ H_1 $ is symmetric, then
Assumption \ref{ass1} implies Assumption \ref{ass2} and
stronger results follow:
\begin{Thm}\label{uniqueness} 
Suppose that Assumption \ref{ass1} holds, and let $H_1$ be a closed symmetric 
operator. Then, $ \overline{U(t,t')} $ is unitary and the following properties hold.
\begin{enumerate}[(i)] 
\item The operator $U(t,t')$ satisfies the following operator equalities:
\begin{align} 
\overline{U(t,t)} =I , \q \overline{U(t,t')} \; \overline{U(t',t'')} = \overline{U(t,t'')} ,
\end{align}
where $ I $ denotes the identity operator.
\item $U(t,t')$ is unique in the following sense. 
If there exist a dense subspace $\widetilde{D}$ in $\H$ and an operator valued function $ V (t,t') \, (t,t' \in \Real ) $ such that 
$\widetilde{D}\subset D(V(t,t'))$ for all $t,t'\in\R$ and for $ \xi \in \widetilde{D} $, $ V(t,t')\xi $ is strongly differentiable with respect to $ t $, and $ V(t,t') \xi \in D(H_1 (t)) $, which satisfies
\begin{align}
V (t,t)\xi =\xi , \q \frac{\partial }{\partial t} V(t,t') \xi = -i H_1 (t) V(t,t') \xi , \q \xi \in \widetilde{D} , \q t,t' \in \R ,
\end{align}
then $V(t,t')\upharpoonright \widetilde{D}$ is closable and $ \overline{V (t,t')\upharpoonright \widetilde{D}} = \overline{U(t,t')} $. In particular, if $ D(V(t,t')) =\H $ and $ V(t,t') $ is bounded for all $ t,t' \in \R $, then $ V(t,t') = \overline{U(t,t')} $. 
\end{enumerate}
\end{Thm}

We discuss the existence of the dynamics generated by $H$.
Let 
\begin{align}
W(t) := e^{-itH_0} \overline{U(t,0)} , \, t\in \R .
\end{align}
\begin{Thm}\label{sch-existence} Suppose that Assumptions \ref{ass1}-\ref{ass2} hold. Then, for each $ \xi \in D(H_0) \cap D $, the vector valued function $ t\mapsto \xi (t) := W(t )\xi $
 is a solution of the initial value problem for the Schr\"{o}dinger equation:
\begin{align}\label{sch30}
\frac{d}{dt} \xi (t) = -iH \xi (t) , \q \xi (0) = \xi .
\end{align}
\end{Thm}

If $H$ is symmetric, we obtain the following 
result:
\begin{Thm}\label{Arai-additional} Suppose that Assumption 2.1 holds and let $H_1$ be a closed symmetric operator.
Then there exists a unique self-adjoint operator $\widetilde H$ such that
\begin{align}
W(t)=e^{-it\widetilde H},\quad t\in \R.\label{Wt}
\end{align}
Moreover
\begin{align}
\overline{U(t,t')}=e^{itH_0}e^{-i(t-t')\widetilde H}e^{-it'H_0},\quad t,t'\in \R\label{Ut}
\end{align}
and
\begin{align}
\overline{H\upharpoonright D\cap D(H_0)} \subset  \widetilde H,\label{H}
\end{align}
where $H\upharpoonright D\cap D(H_0)$ denotes the restriction of $H$ 
to the subspace $D\cap D(H_0)$.
In particular, if $H$ is essentially self-adjoint on $D\cap D(H_0)$, then we have
\begin{align}
\overline{H}=\widetilde{H}.
\end{align}
\end{Thm}


The existence of a solution of the Heisenberg equation \eqref{hei00} is ensured under the following assumptions:
\begin{Ass}\label{ass4} 
\begin{enumerate}[(I)]
\item $ B $ and $ B^* $ are $ A^{1/2} $-bounded and closed.
\item There exists a constant $ b_0 >0 $ such that, for all $L\ge 0$, $\xi\in V_L$
implies $B \xi , B^* \xi \in V_{L+b_0} $.
\end{enumerate}
\end{Ass}
Then, we have
\begin{Thm}\label{scheq-and-heieq} 
Under Assumptions \ref{ass1}, \ref{ass2} and \ref{ass4}, it follows that $ D \subset D(W(-t) B W(t)) $ and the operator valued function $ B(t) $ defined as 
\begin{align}
D(B(t)) := D , \q B(t) \xi & := W(-t) B W(t) \xi , \q \xi \in D, \q t \in \R ,
\end{align}
is a solution of weak Heisenberg equation:
\begin{align}\label{weak-Heisenberg}
\mathrm{w} \text{-} \frac{d}{dt} B(t) = \mathrm{w} \text{-} [iH , B(t) ] \q \text{on} \q D(H_0) \cap D,
\end{align}
where \eqref{weak-Heisenberg} is the abbreviated notation for 
\begin{align}\label{wHei2}
\frac{d}{dt} \left\langle \eta , B(t) \xi \right\rangle = \left\langle (iH)^* \eta , B(t) \xi \right\rangle - \left\langle B(t) ^* \eta , iH \xi \right\rangle , \q  \xi , \eta \in D(H_0) \cap D .
\end{align}
\end{Thm}

Finally, we discuss the existence of a strong solution of the Heisenberg equation \eqref{hei00}.
Let 
\begin{align}
B_0(t) := e^{itH_0} B e^{-itH_0} ,\q t\in \R . 
\end{align}
If the operator valued function $ \R \ni t \mapsto B_0(t) $
satisfies the following conditions, then the stronger result holds.

\begin{Ass}\label{ass5}
\begin{enumerate}[(I)]
\item For each $ \xi \in D(A^{1/2}) $, $ B_0(t) \xi $ is strongly continuously differentiable.
\item The strong derivative of $ B_0(t) $ on $D(A^{1/2})$,
\[ B_0'(t) \xi := \frac{d}{dt} B_0(t) \xi , \q \xi \in D(A^{1/2})  \]
is closable for all $t\in \R$.
\item $ B_0'(t) \; (t\in \R ) $ are uniformly $ A^{1/2} $-bounded. That is, there exist 
constants $ c_0, c_1 \ge 0 $ such that for all $ t \in \R $ and $ \xi \in D(A^{1/2}) $,
\[ \| B_0 ' (t)\xi \| \le c_0 \| A^{1/2} \xi \| + c_1 \| \xi \| . \]
\end{enumerate}
\end{Ass}

The next theorem is concerned with the   
existence of a strong solution of the Heisenberg equation of motion \eqref{hei00}.

\begin{Thm}\label{extra} Under Assumptions \ref{ass1}, \ref{ass2}, \ref{ass4} and \ref{ass5}, it follows that for each $ \xi \in D $, the function $ \R \ni t \mapsto B(t) \xi $ is strongly continuously differentiable, and satisfies
\begin{align}\label{extra5}
\frac{d}{dt} B(t) \xi = W(-t) [iH_1 , B] W(t) \xi + U(0,t) \overline{B_0'(t)} U(t,0)\xi .
\end{align}
Moreover, for each $ \xi \in D(H_0) \cap D $, the equality
\begin{align}
 \frac{d}{dt} B(t) \xi&=  [iH , B(t)] \xi. \label{extra6.5}
\end{align}
holds.
\end{Thm}

\section{Iterative construction of an evolution operator $U(t,t')$}\label{construction}
In this section, we prove Theorem \ref{main-thm1} and Theorem \ref{adjoint}.
\begin{Lem}\label{rel-bdd} Under Assumption \ref{ass1}, $H_1 (t) (A+1)^{-1/2}$ is bounded and there exists a constant $ C \ge 0 $ independent of $t\in \Real$ such that 
\begin{align}
\| H_1 (t) (A+1)^{-1/2} \| \le C  ,\quad t\in \Real.
\end{align}
\end{Lem}
\begin{proof} 
Since $ H_1 $ is $ A^{1/2} $-bounded, there exist constants $ c_0 , c_1 \ge 0 $ satisfying
\begin{align}
\| H_1 \xi \| \le c_0 \| A^{1/2} \xi \| + c_1 \| \xi\| , \q \xi \in D.
\end{align}
Hence, for each $ \xi \in D $, we obtain by operational calculus
\begin{align}
\| H_1 (t )\xi \| & = \| H_1 e^{-itH_0} \xi \| \no \\
& \le c_0\| A^{1/2} e^{-itH_0} \xi \| + c_1 \|  \xi \| \no  \\
& = c_0\| A^{1/2}  \xi \| + c_1 \|  \xi \| \no  .
\end{align}
On the other hand, we have
\[ ||(A+1)^{1/2}\xi||^2=||A^{1/2}\xi||^2+||\xi||^2 .\]
Thus, from the elementary inequality 
\[ (a+b)^2\le 2a^2+2b^2 ,\quad a,b\ge 0 ,\]
we obtain
\begin{align}
\| H_1 (t )\xi \| & \le C \| (A+1)^{1/2} \xi \|  .
\end{align}
This implies that $H_1(t)(A+1)^{-1/2}$ is bounded and 
\[ ||H_1(t)(A+1)^{-1/2} || \le C .\]
\end{proof}
Define
\[ L_\xi := \inf \{ L\ge 0 \, | \, \xi \in V_L \} , \q \xi \in D . \]
For $t,t'\in\Real$, we define a sequence of operators $U_n(t,t')$ for $n=0,1,2,\dots$ in 
the following way:
For $n=0$, put
\begin{align}
D(U_0 (t,t')) = D , \q U_0 (t,t') \xi =\xi , \q \xi \in D .
\end{align}
For $ n\ge 1 $, we inductively define
\begin{align}
D(U_n (t,t')) = D , \q U_n (t,t') \xi = -i \int _{t'} ^t d\tau H_1 (\tau ) U_{n-1} (\tau , t') \xi  , \q \xi \in D , 
\end{align}
where the integration is understood as a strong Riemann integral.
It should be confirmed that $ U_n (t,t') $ is certainly well defined. This
follows from the following lemma: 
\begin{Lem}\label{well-def-U_n}
Suppose that Assumption \ref{ass1} holds. Let $t,t'\in\Real$. Then, for $n= 0,1,2,\dots $, there uniquely exists an operator $U_n(t,t')$ such that 
\begin{enumerate}[(i)]
\item $D(U_n (t,t')) = D$,
\item $U_n(t,t')\xi$ is strongly continuous in $t$,
\item For all $t,t'\in\Real$,
 \[ U_n(t,t')\xi \in V_{L_\xi+nb} ,\]
\item $H_1(t)U_n(t,t')\xi$ is strongly continuous in $t$,
\end{enumerate}
and satisfies the recursion relations
\begin{align}\label{rec-of-U_n}
U_{n+1} (t,t') \xi = -i \int _{t'} ^t d\tau H_1 (\tau ) U_{n} (\tau , t') \xi  , \q \xi \in D , 
\end{align}
for $n=0,1,\dots$, where the integration is a strong Riemann integral.
\end{Lem}
\begin{proof}
If they exist, the uniqueness is obvious by \eqref{rec-of-U_n}.

We prove the existence by induction. Let $n=0$. If we define $U_0(t,t')$ as above
\begin{align}
D(U_0 (t,t')) = D , \q U_0 (t,t') \xi =\xi , \q \xi \in D ,
\end{align}
then (i), (ii), and (iii) clearly hold. To prove (iv), we note that 
\[ H_1(t)\xi = e^{itH_0}H_1(A+1)^{-1/2} e^{-itH_0}(A+1)^{1/2}\xi .\]
The right-hand side is strongly continuous since $H_1(A+1)^{-1/2}$ is bounded. This proves (iv)
in the case where $n=0$.

Suppose that the lemma is true for $n=0,1,2,\dots,k$ for some $k\ge 0$. Then, we can define
$U_{k+1}(t,t')$ via strong Riemann integral as 
\begin{align}
D(U_{k+1} (t,t')) = D , \q U_{k+1} (t,t') \xi = -i \int _{t'} ^t d\tau H_1 (\tau ) U_{k} (\tau , t') \xi  , \q \xi \in D , 
\end{align}
due to (iv). The operator $U_{k+1}(t,t')$ clearly satisfies (i) and (ii). 
Let us prove $U_{k+1}(t,t')$ satisfies (iii). From the induction hypothesis (iii), and
Assumption \ref{ass1} (II), (IV), we find
\[ H_1(t)U_k(t,t')\xi \in V_{L_\xi+(n+1)b}. \]   
Since the subspace $V_{L_\xi+(k+1)b}$ is closed, the strong Riemann integral
\[ \int_{t'}^t  d\tau H_1 (\tau ) U_{k} (\tau , t') \xi \]
also belongs to $V_{L_\xi+(k+1)b}$. This proves (iii) for $n=k+1$. The condition (iv) is
proved as follows.
Since $U_{k+1}(t,t')\xi$ is strongly differentiable with respect to $t$,
it is strongly continuous. 
On the other hand, the map
\[ t\mapsto e^{itH_0} H_1(A+1)^{-1/2} e^{-itH_0}(A+1)^{1/2} \]
is continuous in the strong operator topology.
Thus, we have for $h\in\Real$
\begin{align}
&|| H_1(t+h)U_{k+1}(t+h,t')\xi - H_1(t)U_{k+1}(t,t')\xi || \no \\
&\quad  \le 
|| (e^{i(t+h)H_0} H_1(A+1)^{-1/2} e^{-i(t+h)H_0} - e^{itH_0} H_1(A+1)^{-1/2} e^{-itH_0})(A+1)^{1/2}U_{k+1}(t,t')\xi || \no\\
&\quad \q+ ||H_1(A+1)^{-1/2} ||\, || (A+1)^{1/2}(U_{k+1}(t+h,t')-U_{k+1}(t,t'))\xi ||,
\end{align}
which shows (iv).
The formula \eqref{rec-of-U_n} holds by the construction. Therefore,
the lemma remains true for $n=k+1$ and this completes the proof.
\end{proof}
\begin{Lem}\label{ess-estimate}
Suppose that Assumption \ref{ass1} holds. Let $t,t'\in\Real$ and $\xi\in D$. Then, we can estimate
\begin{align}\label{estimate4}
\| U_n (t,t') \xi \| \le \frac{|t-t'| ^n }{n!} C^n \big( L_\xi + (n-1) b+1\big) ^{1/2}   \cdots \big(  L_\xi +1\big)^{1/2}   \| \xi \| ,
\end{align}
for $n=0,1,2,\dots$.
\end{Lem}
\begin{proof}
We prove by induction.
If $ n=0 $, this is obvious. Assume that \eqref{estimate4} holds for some $ n\ge 0 $. 
Then, if $t'\le t$, we have
\begin{align}
\| U_{n+1} (t,t') \xi  \| & \le \int _{t'} ^t d\tau \| H_1 (\tau ) U_n (\tau , t') \xi \| \no \\
&= \int _{t'} ^t d\tau \| H_1 (\tau )(A+1)^{-1/2} ||\,||(A+1)^{1/2} U_n (\tau , t') \xi \| \no \\
& \le \int _{t'} ^t d\tau\,C\cdot \big( L_\xi + nb+1\big) ^{1/2}  \frac{|\tau - t'| ^{n}}{n!}C^n \big( L_\xi + (n-1) b+1\big) ^{1/2}   \cdots \big(  L_\xi +1\big)^{1/2}   \| \xi \| \no  \\
& \le \frac{|t-t'| ^{n+1} }{(n+1)!} C^{n+1} \big( L_\xi + nb+1\big) ^{1/2}   \cdots \big(  L_\xi +1\big)^{1/2}   \| \xi \| .
\end{align}
By a similar computation, one finds that this estimate is also true in the case where $t<t'$. 
Thus the induction completes and we obtain \eqref{estimate4}.
\end{proof}
\begin{Lem}\label{U-estimate} Under Assumption \ref{ass1}, for all $ t,t' \in \Real $ and $ \xi \in D $, the followings hold.
\begin{align}
\label{estimate1} & \sum _{n=0} ^\infty \| U_n (t,t') \xi \| <\infty ,\\
\label{estimate2} & \sum _{n=0} ^\infty \| H_1 (t) U_n (t,t') \xi \| <\infty , \\
\label{estimate3} & \sum _{n=0} ^\infty \| U_n (t,t') H_1 (t') \xi \| <\infty ,
\end{align}
Furthermore, these convergences are uniform in $(t,t')$ on any compact subset in $\Real^2$.
\end{Lem}

\begin{proof} 
From Lemma \ref{ess-estimate}, we know
\begin{align}
\sum _{n=0} ^\infty \| U_n (t,t') \xi \| \le \sum _{n=0} ^\infty \frac{|t-t'| ^n }{n!} C^n \big( L_\xi + (n-1) b+1\big) ^{1/2}   \cdots \big(  L_\xi +1\big)^{1/2}   \| \xi \|  .\label{ratio-test-U}
\end{align}
Let $a_n(t,t') $ be the $n$-th term of the summation in the right-hand side of \eqref{ratio-test-U}. One can see that 
\[ \lim_{n\to\infty}\left|\frac{a_{n+1}(t,t')}{a_n(t,t')}\right| = 0, \]
uniformly in $(t,t')$ on any compact subset in the plane.
By using d'Alembert's ratio test, the right hand side converges
uniformly in $(t,t')$ on any compact subset, and obtain \eqref{estimate1}.

The convergence of the other two series' \eqref{estimate2} and \eqref{estimate3} are also proved in a similar way,
and we omit the proof. 
\end{proof}

\begin{Lem}\label{multi-sconti}
Suppose that Assumption \ref{ass1} holds. Let $\xi \in D$ and $n=1,2,\dots$. Then, the $n$-variable function 
from $\Real^n$ into $\mathcal{H}$
\[ \Real^n\ni (t_1,\dots,t_n)\mapsto H_1(t_1)\dots H_1(t_n)\xi \in D \]
is continuous on $\Real^n$ with respect to the usual topology in $\Real^n$ and the strong topology in $\mathcal{H}$.
\end{Lem} 

\begin{proof}
Fix $\xi \in D$.
We prove by induction with respect to $n$.

Set $n=1$. We will prove 
\[ t\mapsto H_1(t)\xi \]
is strongly continuous. But, this has already been proved in the proof of Lemma \ref{well-def-U_n}.

Suppose that the assertion is valid for some $n\ge 1$. We prove the $(n+1)$-variable function
\[  (t_1,\dots,t_{n+1})\mapsto H_1(t_1)\dots H_1(t_{n+1})\xi \] 
is strongly continuous at any $(t_1,\dots,t_{n+1} )\in\Real^{n+1}$.
We use the abbreviated notations such as
\[ \bvec t=(t_2,\dots,t_{n+1})\in\Real^n,\quad (t_1,\bvec t)=(t_1,\dots,t_{n+1})\in \Real^{n+1},  \]
and $|\bvec t - \bvec s|$ denotes the standard Euclidean distance. 
Choose arbitrary $\epsilon > 0$. By the induction hypothesis, there is a $\delta(\bvec t,\epsilon)>0$
such that for all $\bvec s=(s_2,\dots,s_{n+1})\in\Real^n$ with $|\bvec t -\bvec s | <\delta(\bvec t,\epsilon)$,
\[ ||( H_1(t_2)\dots H_1(t_{n+1}) -H_1(s_2)\dots H_1(s_{n+1}))\xi || < \frac{\epsilon}{2C(L_\xi+ nb+1)^{1/2}} .\]
On the other hand, since the mapping from $\Real$ to the set of bounded linear operators
in $\mathcal{H}$
 \[ t\mapsto e^{itH_0} H_1(A+1)^{-1/2} e^{-itH_0} \] 
 is strongly continuous, there is a $\delta'(t_1,\bvec t,\epsilon)>0$ such that for all $s_1\in\Real$ with
 $|t_1-s_1|<\delta'(t_1,\bvec t,\epsilon)$
 \begin{align}
 || (H_1(t_1)-H_1(s_1)) H_1(t_2)\dots H_1(t_{n+1})\xi || 
 <\frac{\epsilon}{2},
 \end{align}
 because for any $\psi\in D(A^{1/2})$, we have
 \[ (H_1(t_1)-H_1(s_1))\psi =(e^{it_1H_0}H_1(A+1)^{-1/2} e^{-it_1H_0}-
 e^{is_1H_0}H_1(A+1)^{-1/2} e^{-is_1H_0})\cdot (A+1)^{1/2}\psi .\]
 From these estimates, one finds that for all $(s_1,\bvec s)\in\Real^{n+1}$ with 
 $|(t_1,\bvec t) - (s_1,\bvec s)| < \mathrm{min}(\delta(\bvec t,\epsilon),\delta'(t_1,\bvec t,\epsilon)) $
\begin{align}
&||H_1(t_1)\dots H_1(t_{n+1})\xi -H_1(s_1)\dots H_1(s_{n+1})\xi || \no\\
&\le || (H_1(t_1)-H_1(s_1)) H_1(t_2)\dots H_1(t_{n+1})\xi || +\no\\ 
&\quad +|| H_1(s_1)( H_1(t_2)\dots H_1(t_{n+1}) -H_1(s_2)\dots H_1(s_{n+1}))\xi || \no \\
&\le || (H_1(t_1)-H_1(s_1)) H_1(t_2)\dots H_1(t_{n+1})\xi || + \no\\
&\quad +C(L_\xi+ nb+1)^{1/2}||( H_1(t_2)\dots H_1(t_{n+1}) -H_1(s_2)\dots H_1(s_{n+1}))\xi ||\no\\
&<\frac{\epsilon}{2}+\frac{\epsilon}{2C(L_\xi+ nb+1)^{1/2}}\cdot C(L_\xi+ nb+1)^{1/2}=\epsilon , 
\end{align}
where we have used the fact that the vector $(H_1(t_2)\dots H_1(t_{n+1}) -H_1(s_2)\dots H_1(s_{n+1}))\xi $
belongs to $V_{L_\xi+nb}$. This proves the lemma.
\end{proof}

From Lemma \ref{multi-sconti}, we can define a strong Bochner integral 
in $\Real^n$ 
\begin{align}\label{boch-int}
\int_A d^n \tau \,H_1(\tau_1)\dots H_1(\tau_n)\xi , \quad \xi \in D, 
\end{align}
for any Borel measurable bounded subset $A\subset \Real^n$, where $ d^n \tau $ denotes $ n $-dimensional Lebesgue measure. In particular, since $H_1$ is
closed,
we obtain for $t'\le t$ and $\xi\in D$,
\begin{align}
\int_{t\ge \tau_1\ge \dots \ge \tau_n \ge t'} d^n \tau \,H_1(\tau_1)\dots H_1(\tau_n)\xi
&=\int_{t'}^t d\tau_1\,\int_{t'}^{\tau_1} d\tau_2\,\dots\int_{t'}^{\tau_{n-1}} d\tau_n\,H_1(\tau_1)H_1(\tau_2)\dots H_1(\tau_n)\xi \no\\
&=\int_{t'}^t d\tau_1\,H_1(\tau_1)\int_{t'}^{\tau_1} d\tau_2\,H_1(\tau_2)\dots\int_{t'}^{\tau_{n-1}} d\tau_n\,H_1(\tau_n)\xi\no\\
&=\int_{t'}^t d\tau_1\,H_1(\tau_1)\dots \int_{t'}^{\tau_{n-2}}i H_1(\tau_{n-1})U_1(\tau_{n-1},t')\xi\no\\
&=\dots 
=i^nU_{n}(t,t')\xi,\label{time-order-of-U_n}
\end{align}
where \eqref{rec-of-U_n} was used in the third equality.

\begin{proof}[Proof of Theorem \ref{main-thm1}]
Let $\xi \in D$ and define
\[S_n (t,t') \xi := \sum _{j=0} ^n U_j(t,t')\xi , \q n\ge 0 , \q t,t' \in \R.\]
It is clear from Lemma \ref{U-estimate} \eqref{estimate1} that $\{S_n(t,t')\xi \}_n$ is Cauchy in $\mathcal{H}$. Thus, we can define 
\begin{align}
D(U(t,t')) = D , \q U(t,t') \xi = \lim_{n\to\infty} S_n(t,t')\xi =\sum _{n=0} ^\infty U_n (t,t') \xi , \q \xi \in D , \q t,t' \in \Real .
\end{align}
This infinite summation converges absolutely and uniformly in $(t,t')$ on any compact set 
$ K \subset \R ^2 $. Note that $U_k(t,t')\xi $ ($k=0,1,\dots ,n$) are strongly differentiable with respect to $t$, so is $S_n(t,t')\xi $. The derivative of $S_n(t,t')$ with respect to $t$ becomes
\begin{align}
\pdfrac{S_n(t,t')\xi}{t} &= \sum _{j=0} ^n \pdfrac{U_j(t,t')\xi }{t} \no\\
 &=-i \sum _{j=1} ^n H_1(t)U_{j-1}(t,t')\xi \no\\
 &=-i H_1(t) \sum _{j=0} ^{n-1} U_{j}(t,t')\xi\no\\
 &= -iH_1(t) S_{n-1}(t,t') \xi .
\end{align}
By \eqref{estimate2}, one finds that $\{ H_1(t)S_n(t,t')\xi \}_n$ is --- and therefore
$\{ (\partial /\partial t) S_n(t,t')\xi \}_n$ is --- Cauchy.
Hence, the limit
\[ \lim_{n\to\infty} \pdfrac{S_n(t,t')\xi}{t} = -i \lim_{n\to\infty } H_1(t) S_{n-1}(t,t') \]
exists. Due to the fact that $ H_1 (t) $ is closed, this implies $ U(t,t') \xi \in D(H_1 (t)) $ and
\begin{align}
\frac{\partial }{\partial t} S_n (t,t') \xi \rightarrow -i H_1(t) U(t,t') \xi , \q (n\rightarrow \infty ) ,
\end{align}
uniformly in $(t,t')$ on any compact set $ K \subset \R ^2 $.
Since the function $ t \mapsto (\partial/\partial t) \,S_n (t,t') \xi $ is strongly continuous, so
is its uniform limit $ -i H_1(t) U(t,t') \xi $. Then, by exchanging limit and integration, we have
\[ \lim_{n\to\infty}\int _{t'} ^t d\tau \frac{\partial }{\partial \tau } S_n (\tau ,t') \xi = -i \int _{t'} ^t d\tau H_1(\tau ) U(\tau ,t') \xi, \] 
where the convergence is uniform on $ K $. Since
\begin{align}
\int _{t'} ^t d\tau \frac{\partial }{\partial \tau } S_n (\tau ,t') \xi = S_n (t,t') \xi - \xi \rightarrow U(t,t') \xi - \xi , \q (n\rightarrow \infty ) ,
\end{align}
we have
\begin{align}
U(t,t') \xi = \xi -i \int _{t'} ^t d\tau H_1(\tau ) U(\tau ,t') \xi  ,
\end{align}
which implies that $ U(t,t') \xi $ is strongly continuously differentiable with respect to $t$ at all $ (t,t') \in K $. Since $ K $ is arbitrary, one concludes that \eqref{DE1} holds.

Next, we prove \eqref{DE2}. Let $t'\le t$. By interchanging the order of 
integrations, we have from \eqref{time-order-of-U_n}
\begin{align}
U_n (t,t') \xi & = (-i)^n \int _{t' \le \tau _n \le \cdots \le \tau _2 \le \tau _1 \le t} d^n \tau \, H_1 (\tau _1) H_1 (\tau _2) \cdots H_1 (\tau _n) \xi \no\\
& = (-i)^n \int _{t'} ^t d\tau _n \int _{\tau _n } ^t d\tau _{n-1} \cdots \int _{\tau _2} ^t d\tau _1 \, H_1 (\tau _1) \cdots H_1 (\tau _{n-1} ) H_1 (\tau _n) \no \xi \\
& = i^n \int _t ^{t'} d\tau _n \int _t ^{\tau _n} d\tau _{n-1} \cdots \int _t ^{\tau _2} d\tau _1 H_1 (\tau _1) \cdots H_1 (\tau _{n-1}) H_1 (\tau _n) \xi ,
\end{align}
and this implies
\begin{align}\label{inverse-rec}
U_{n+1} (t,t') \xi = i \int _t ^{t'} d\tau U_n (t,\tau ) H_1 (\tau ) \xi .
\end{align}
One can check in the same manner that \eqref{inverse-rec} remains valid even if 
$t<t'$. Hence we find that $ U_n (t,t') \xi  $ is differentiable with respect to $ t' $, and 
\begin{align}\label{t'diff}
\frac{\partial }{\partial t'} U_{n+1}(t,t') \xi = i U_n (t,t') H_1 (t') \xi .
\end{align}
By using \eqref{estimate3}, we can repeat a discussion similar to the one 
in the previous paragraph to obtain \eqref{DE1}, to learn that $U(t,t')\xi $ is strongly continuously differentiable with respect to
 $t'$ at all $ (t,t') \in \Real^2 $, and satisfies \eqref{DE2}.
\end{proof}

In the rest of the present section, we also use Assumption \ref{ass2} in order to obtain 
more detailed results. 
We can derive the following lemmas in the same manner as before.

\begin{Lem}\label{rel-bdd2} Under Assumption \ref{ass2}, $H_1 (t)^* (A+1)^{-1/2}$ is bounded and there exists a constant $ C'\ge 0 $ independent of $t\in \Real$ such that 
\begin{align}
\| H_1 (t)^* (A+1)^{-1/2} \| \le C'  ,\quad t\in \Real.
\end{align}
\end{Lem}

\begin{Lem}\label{multi-sconti^*}
Under Assumption \ref{ass2},  for all $\xi \in D$, the $n$-variable function 
\[ \Real^n\ni (t_1,\dots,t_n)\mapsto H_1(t_1)^* \dots H_1(t_n)^* \xi \in D \]
is strongly continuous on $\Real^n$.
\end{Lem}
Lemma \ref{multi-sconti^*} ensures the existence of a strong Bochner integral 
\begin{align}
\int_A d^n \tau \,H_1(\tau_1)^*\dots H_1(\tau_n)^*\xi , \quad \xi \in D, 
\end{align}
for any bounded Borel set $A\subset\R^n$
and allows us to perform computations such as \eqref{time-order-of-U_n}
with $H_1$ replaced by $H_1^*$. 
 
\begin{Lem}\label{int-rep-lem}
Let Assumption \ref{ass1} and Assumption \ref{ass2} hold. Then, $D\subset D(U_n(t,t')^*)$
 and for all $\xi\in D$,
\begin{align}
U_n(t,t')^*\xi & = i^n \int _{t'} ^t d\tau _1 \int _{t'} ^{\tau _1} d\tau _2 \cdots \int _{t'} ^{\tau _{n-1}} d\tau _n \, H_1 (\tau _n)^* \cdots H_1 (\tau _2) ^* H_1 (\tau _1) ^* \xi \label{U_n^*-int-rep} \no \\
& = (-i)^n \int _{t} ^{t'} d\tau _n \int _{t} ^{\tau _n} d\tau _{n-1} \cdots \int _{t} ^{\tau _2} d\tau _1 \, H_1 (\tau _n)^* H_1 (\tau _{n-1}) ^* \cdots H_1 (\tau _1) ^* \xi . 
\end{align}
In particular, $ U_n(t,t') ^* \xi $ is strongly continuously differentiable with respect to $ t $ and $ t' $.
\end{Lem}

\begin{proof}
Choose arbitrary $\xi,\eta \in D$. Then
\begin{align}
\Expect{U_n(t,t')\eta,\xi} &= i^n\int_{t'}^t d\tau_1 \dots \int_{t'}^{\tau_{n-1}} d\tau_{n} \Expect{H_1(\tau_1)\dots H_1(\tau_n)\eta,\xi} \no\\
 &= i^n\int_{t'}^t d\tau_1 \dots \int_{t'}^{\tau_{n-1}} d\tau_{n} \Expect{\eta,H_1(\tau_n)^*\dots H_1(\tau_1)^*\xi} \no\\
 &=\Expect{\eta,i^n\int_{t'}^t d\tau_1 \dots \int_{t'}^{\tau_{n-1}} d\tau_{n}\, H_1(\tau_n)^*\dots H_1(\tau_1)^*\xi}. 
\end{align}
This means $\xi\in D(U_n(t,t')^*)$ and 
\eqref{U_n^*-int-rep}.
\end{proof}
 
From Lemma \ref{rel-bdd2}, and Lemma \ref{int-rep-lem}, we can derive the following estimation 
 for $\| U_n (t,t')^* \xi \|$, whose proof will be omitted since it is very similar to that of Lemma
 \ref{ess-estimate}.
\begin{Lem}\label{ess-estimate^*}
Let Assumption \ref{ass1} and Assumption \ref{ass2} hold. Let $t,t'\in\Real$ and $\xi\in D$. Then, the estimate
\begin{align}\label{estimate4^*}
\| U_n (t,t')^* \xi \| \le \frac{|t-t'| ^n }{n!} C'^n \big( L_\xi + (n-1) b+1\big) ^{1/2}   \cdots \big(  L_\xi +1\big)^{1/2}   \| \xi \| ,
\end{align}
holds for $n=0,1,2,\dots$.
\end{Lem}

Once this estimate is obtained, the corresponding statement to Lemma \ref{U-estimate}
is proved:
\begin{Lem}\label{U-estimate^*}
Under Assumptions \ref{ass1}-\ref{ass2}, the following summations converge
uniformly in $(t,t')$ on any compact set in the plane.
\begin{align}
\label{adjoint1} & \sum _{n=0} ^\infty \| U_n (t,t') ^* \xi \| < \infty , \\
\label{adjoint2} & \sum _{n=0} ^\infty \| U_n (t,t') ^* H_1(t) ^* \xi \| < \infty , \\
\label{adjoint3} & \sum _{n=0} ^\infty \| H_1 (t') ^* U_n (t,t') ^* \xi \| < \infty .
\end{align}
\end{Lem}

\begin{proof}[Proof of Theorem \ref{adjoint}] 
From Lemma \ref{U-estimate^*} \eqref{adjoint1}, one finds
\begin{align}
\sum_{n=0}^N U_n(t,t')^*\xi,\q \xi\in D
\end{align}
absolutely converges uniformly in $(t,t')$ on any compact set.
For all $ \xi , \eta \in D $, we obtain
\begin{align}
\left\langle \eta , U(t,t') \xi \right\rangle & = \sum _{n=0} ^\infty \left\langle \eta , U_n (t,t') \xi \right\rangle \no \\
& = \sum _{n=0} ^\infty \left\langle U_n (t,t') ^* \eta ,  \xi \right\rangle \no \\
& = \left\langle \sum _{n=0} ^\infty U_n (t,t') ^* \eta ,  \xi \right\rangle ,
\end{align}
since $ \eta \in D (U_n(t,t') ^*) $ for all $ n $. Thus, we obtain $ D \subset D(U(t,t') ^*) $ and
\begin{align}\label{expansion-of-U^*}
U(t,t') ^* \xi = \sum _{n=0} ^\infty U_n (t,t') ^* \xi , \q \xi \in D .
\end{align}

From \eqref{expansion-of-U^*}, we can mimic the proof of Theorem \ref{main-thm1} by
using \eqref{adjoint2} and \eqref{adjoint3}, to
obtain \eqref{adjoint4} and \eqref{adjoint5}.
\end{proof}

\section{Properties of time evolution operator}\label{properties}
In the present section, we prove several properties of the time evolution operator $U(t,t')$, and prove Theorems \ref{main-thm2} and
\ref{uniqueness}.

\begin{Prop} Suppose that Assumptions \ref{ass1} and \ref{ass2} hold. If $ \xi \in D $, then for each $ t,t', s,s' \in \Real $, $ U(s,s') \xi \in D(\overline{U(t,t')}) $ and
\begin{align}\label{product}
\overline{U(t,t')} U(s,s') \xi = \sum _{m,n =0} ^\infty U_m (t,t') U_n (s,s') \xi ,
\end{align}
where the right hand side converges absolutely, and does not depend upon the summation order.
\end{Prop}

\begin{proof}
For all $ \xi \in D $ and all $(t,t'),(s,s')\in\Real^2$, it is clear that $ S_n (s,s') \xi \in D(U (t,t')) $. 
Since $S_n(s,s')\xi$ converges to $U(t,t')\xi$ as $n$ tends to infinity, it suffices to prove
that
$U(t,t')S_n(s,s')\xi$ converges as $n\to\infty$.
We have already know that
\[ U(t,t')S_n(s,s')\xi =\sum _{m=0} ^\infty \sum _{j=0} ^n  U_m (t,t') U_j (s,s') \xi , \]
therefore, it is sufficient to derive
\[ \sum _{m=0} ^\infty \sum _{j=0} ^\infty \| U_m (t,t') U_j (s,s') \xi \| <\infty . \] 
By using \eqref{estimate4},
\begin{align}
&\sum _{m=0} ^\infty \sum _{j=0} ^\infty \| U_m (t,t') U_j (s,s') \xi \|\no\\
&\quad \le \sum _{m=0} ^\infty \sum _{j=0} ^\infty \frac{|t-t'| ^m |s-s'| ^j}{m!j!} C^{m+j}(L_\xi +(m+j-1) b+1)^{1/2}  \cdots ( L_\xi  + 1 ) ^{1/2}\| \xi \| \no\\
&\quad = \sum_{N=0}^\infty \sum_{m=0}^N 
\frac{|t-t'| ^m |s-s'| ^{N-m}}{m!(N-m)!} C^{N}(L_\xi +(N-1) b+1)^{1/2}  \cdots ( L_\xi  + 1 ) ^{1/2}\| \xi \|\no\\
&\quad = \sum_{N=0}^\infty \frac{1}{N!}
\left(C(|t-t'| + |s-s'|)\right)^N (L_\xi +(N-1) b+1)^{1/2}  \cdots ( L_\xi  + 1 ) ^{1/2}\| \xi \|
\end{align}
From the d'Alembert's ratio test, this is finite, which proves \eqref{product}.
\end{proof}

\begin{proof}[Proof of Theorem \ref{main-thm2}] 
We first prove (i). Take arbitrary $\xi,\eta\in D$. $ U(t,t) \xi =\xi $ is obvious. By Theorems \ref{main-thm1} and \ref{adjoint}, the function $ t' \mapsto \Expect{ \eta , \overline{U(t,t')} U(t' ,t'') \xi } =\Expect{ U(t,t')^*\eta , U(t' ,t'') \xi}$ is differentiable and
\begin{align}
& \frac{\partial }{\partial t'} \left\langle \eta , \overline{U(t,t')} U(t' , t'') \xi \right\rangle \no \\
 &\quad = \left\langle -i H_1 (t' ) ^* U (t,t') ^* \eta , U(t' ,t'') \xi \right\rangle + \left\langle U (t,t') ^* \eta , -i H_1 (t' ) U(t' ,t'') \xi \right\rangle\no  \\
&\quad =  0 .
\end{align}
Thus $\Expect{ \eta , \overline{U(t,t')} U(t' ,t'') \xi }$ is independent of $t'$, which implies
\begin{align}
\left\langle \eta , \overline{U(t,t')} U(t' , t'') \xi \right\rangle &= \left\langle \eta , \overline{U(t,t'')} U(t'',t'') \xi \right\rangle \no\\
 &= \left\langle \eta , U(t,t'') \xi \right\rangle .
\end{align}
Since $ \eta \in D $ is arbitrary, it follows that 
\begin{align}
\overline{U(t,t')} U(t' , t'') \xi=U(t,t'') \xi.
\end{align}
Hence, we obtain \eqref{associative} because $\xi\in D$ is arbitrary and 
$D=D(\overline{U(t,t')} U(t' , t''))$.

Next, we prove (ii). Observe $ e^{isH_0} H_1(t) e^{-isH_0} = H_1(t +s) $ by definition. Suppose $ t \ge t' $. Then, for each $ n \in \Natural , \, \xi \in D $,
we obtain
\begin{align}
& e^{isH_0} U_n(t,t') e^{-isH_0} \xi \no \\
= & \int_{t\ge \tau_1\ge \dots \ge \tau_n \ge t'} d^n \tau \, H_1(\tau_1+s ) \dots H_1(\tau_n +s )\xi \no \\
= & \int_{t+s\ge \tau_1\ge \dots \ge \tau_n \ge t'+s} d^n \tau \, H_1(\tau_1 ) \dots H_1(\tau_n  )\xi \no \\
\label{parallel1} = & U_n (t+s, t'+s) \xi .
\end{align}
The relation \eqref{parallel1} remains valid in the case where $ t<t' $. Thus we have for all $ (t,t') \in \R ^2 $,
\begin{align*}
e^{isH_0} U(t,t') e^{-isH_0} \xi & = \sum _{n=0} ^\infty e^{isH_0} U_n(t,t') e^{-isH_0} \xi \\
& = \sum _{n=0} ^\infty U_n(t+s, t'+s) \xi \\
& = U(t+s, t'+s) \xi .
\end{align*}
Since $ D $ is common core of $ \overline{U(t,t')} $ and $ \overline{U(t+s, t'+s)} $, we obtain the desired result.
\end{proof}

\begin{proof}[Proof of Theorem \ref{uniqueness}] 
We first prove that under the present situation, Assumption \ref{ass1} implies Assumption \ref{ass2}.
Let $H_1$ be symmetric and Assumption \ref{ass1} hold.
Then, by the fact that $H_1$ is symmetric and $A^{1/2}$- bounded, we find 
\begin{align}
D(A^{1/2})\subset D(H_1) \subset D(H_1^*),  
\end{align}
which means that $H_1^*$ is also $A^{1/2}$- bounded. Moreover, for each $\xi\in V_L$, one finds
\begin{align}
H_1^* \xi = H_1\xi \in V_{L+b}.
\end{align}
Thus, Assumption \ref{ass2} is satisfied.

Next, we prove the unitarity. Since $H_1$ is symmetric, 
one obtains for all $\xi \in D$
\begin{align}
\frac{\partial}{\partial t} \| U(t,t') \xi \| ^2 = & \left\langle -i H_1 (t) U(t,t') \xi , U(t,t') \xi \right\rangle + \left\langle U(t,t') \xi , -i H_1 (t) U(t,t') \xi \right\rangle\no \\
= & 0 .
\end{align}
Therefore, $ \overline{U(t,t')} $ is isometry, in particular, bounded. By using Theorem \ref{main-thm2}, one finds the operator equality
\begin{align}
\overline{U(t,t')} \; \overline{U(t',t)} = I , \q t,t' \in \R ,
\end{align}
which implies that $ U(t,t') $ is surjective. Hence, it is unitary.

The statement (i) is directly follows from Theorem \ref{main-thm2}. 

We prove (ii). For each $ \xi \in D, \,  \eta \in \widetilde{D} $ and $ t\in \R $, we have
\begin{align}
\frac{\partial}{\partial t}  \left\langle V (t,t') \eta , U(t,t') \xi \right\rangle & = \left\langle -iH_1 (t) V (t,t') \eta , U(t,t') \xi \right\rangle + \left\langle V (t,t') \eta , -iH_1 (t) U(t,t') \xi \right\rangle \no \\
& = 0 .
\end{align}
Thus, we obtain
\begin{align}\label{additional7}
\left\langle V (t,t') \eta , U(t,t') \xi \right\rangle = \left\langle \eta , \xi \right\rangle , \q \xi \in D, \; \eta \in \widetilde{D} . 
\end{align}
Since $ D $ is dense in $\H$ and since $\overline{U(t,t')}$ is unitary and satisfies $\overline{U(t,t')}\,\overline{U(t',t)}=I$, \eqref{additional7} yields for all $\eta\in \widetilde{D}$, 
\begin{align}\label{above}
 V (t,t') \eta = \overline{U(t,t')} \eta ,\quad \eta\in \widetilde{D}.
 \end{align}
 Suppose that a sequence $\{\eta_n\}_n\subset \widetilde{D}$ satisfies that $\eta_n \to 0$ as $n$ tends to infinity. 
 Then, \eqref{above} shows that $V(t,t')\eta_n$ converges to $0$,
 which means that $V(t,t')$ is closable.
 Take arbitrary $\psi\in\H$. Then, there is a sequence $\{\eta_n\}_n$ which converges to $\psi$ as $n\to\infty$, 
since $ \widetilde{D} $ is dense in $\H$. Then, \eqref{above} implies $\psi\in D(\overline{V(t,t')\upharpoonright \widetilde{D}})$ and
$ \overline{V (t,t')\upharpoonright \widetilde{D}}\psi = \overline{U(t,t')} \psi$ for all $\psi\in \H$.
\end{proof}

\section{Schr\"odinger and Heisenberg equations of motion}\label{solutions}
In this section, we construct solutions of the Schr\"odinger and Heisenberg equations of motion
via the time evolution operator $U(t,t')$, and prove Theorems \ref{sch-existence} and \ref{scheq-and-heieq}.
Throughout this section, we use Assumptions \ref{ass1} and \ref{ass2}.
Hereafter, we denote the closure of $ U(t,t') $ by the same symbol. 
Recall that $ W(t) = e^{-itH_0} U(t,0) , \, t\in \R $. Put 
\[ D':=D\cap D(H_0). \]
We remark that
$D'$ is dense in $\H$ under Assumption \ref{ass1} (I) and (II). This can be seen as follows.
Let $\psi\in D(H_0)$ and, for $n\in\Natural$,
\[ \psi_n:=E_A([0,n])\psi. \]
Then $\psi\in D$. By Assumption \ref{ass1} (II), $\psi_n\in D(H_0)$. Hence, $\psi_n\in D'$.
It is clear that $\psi_n\to\psi$ as $n$ tends to infinity. Thus $D'$ is dense in $D(H_0)$ and
then also in $\H$.

\begin{proof}[Proof of Theorem \ref{sch-existence}] For all $ \eta \in D(H_0) $, 
\begin{align}\label{sch1}
\frac{d}{dt} \left\langle \eta , W(t) \xi \right\rangle & = \frac{d}{dt} \left\langle e^{itH_0} \eta , U(t,0) \xi \right\rangle \no \\
& = \left\langle iH_0 \eta , W(t) \xi \right\rangle + \left\langle \eta , -iH_1 W(t) \xi \right\rangle .  
\end{align}
By Theorem \ref{main-thm2} (ii), $ W(t) $ can be rewritten as $ U(0, -t) e^{-itH_0} $. Since, for all $ \xi \in D' $, the functions $ e^{-itH_0} \xi $ and $ U(0,-t) \xi $ are strongly differentiable and since $ H_0 D' \subset D $, it follows that the  function $ W(t) \xi  $ is also strongly differentiable and the derivative becomes
\begin{align}\label{sch2}
\frac{d}{dt} W(t) \xi & = U(0, -t) \big( -iH_1(-t) - iH_0 \big) e^{-itH_0} \xi \no \\
& = W(t) (-iH) \xi .
\end{align}
Hence, by \eqref{sch1} and \eqref{sch2}, we have
\begin{align}
\left\langle iH_0 \eta , W(t) \xi \right\rangle + \left\langle \eta , -iH_1 W(t) \xi \right\rangle = \left\langle \eta , W(t) (-iH) \xi \right\rangle , \q \eta \in D(H_0) ,
\end{align}
which implies that $ W(t) \xi \in D(H_0) $
\begin{align} \label{WH}
-i W(t) H\xi = -iH W(t) \xi 
\end{align}
Therefore, we obtain \eqref{sch30}.
\end{proof}

\begin{proof}[Proof of Theorem \ref{Arai-additional}]. By Theorems \ref{main-thm1} and \ref{uniqueness}, for all $t\in \R$, $W(t)$ is unitary with $W(0)=I$ and strongly continuous in $t\in \R$. 
In the present case, Assumption \ref{ass2} holds too. Hence, by Theorem \ref{main-thm2} (ii), we have
$$
W(t)W(s)=W(t+s), \quad s,t\in \R.
$$ 
Thus $\{W(t)\}_{t\in \R}$ is a strongly continuous one-parameter unitary  group.
Hence, by Stone's theorem, 
the first statement  of the theorem holds. 

By (\ref{Wt}), we have for all $t\in \R$
$$
\overline{U(t,0)}=e^{itH_0}e^{-it\widetilde H}.
$$
By this equation and Theorem \ref{uniqueness} (i), we obtain (\ref{Ut}).

It follows from Theorem \ref{sch-existence} that $D'=D\cap D(H_0)\subset D(\widetilde H)$ and $H\xi=\widetilde H\xi$, $\xi\in D '$. Hence (\ref{H}) follows. 

If $H$ is essentially self-adjoint on the subspace $D'$, then one finds
\[ \overline{H}\subset \widetilde{H}. \]
But since both $\overline{H}$ and $\widetilde{H}$ are self-adjoint, we have the equality. 
\end{proof}

Next, we prepare some lemmas to prove Theorem \ref{scheq-and-heieq}.

\begin{Lem}\label{rel-bdd^*} Under Assumption \ref{ass4}, $ B (A+1)^{-1/2} $ and $ B^* (A+1)^{-1/2} $ are bounded and there exists a constant $ C_0 \ge 0 $ such that 
\begin{align}\label{Bestimate}
\| B (A+1)^{-1/2} \| , \, \| B^* (A+1)^{-1/2} \| \le C_0  . 
\end{align}
\end{Lem}
\begin{proof}
This can be proved in the same way as Lemma \ref{rel-bdd}
\end{proof}

\begin{Lem}\label{heiLem} Under Assumptions \ref{ass1}, \ref{ass2} and \ref{ass4}, the followings hold.
\begin{enumerate}[(i)]
\item For all $ \xi \in D \cap D(H_0) $, the function $ W(t) ^* \xi $ is strongly differentiable and satisfies
\begin{align}\label{sch8}
\frac{d}{dt} W(t) ^* \xi = iH^* W(t) ^* \xi = i W(t) ^* H^* \xi .
\end{align}

\item $ D \subset D(W(-t) B W(t)) $.

\item $ D \subset D(W(t) ^* B^* W(-t) ^*) $ and 
\begin{align}\label{B*}
B(t)^* \xi = W(t) ^* B^* W(-t) ^* \xi , \q \xi \in D ,
\end{align}
hold.

\item For all $ \xi \in D' $, the function $ B W(t) \xi $ is strongly differentiable and satisfies
\begin{align}\label{BWdiff}
\frac{d}{dt} B W(t) \xi = -i BW(t) H \xi .
\end{align}
\end{enumerate}
\end{Lem}

\begin{proof} 
\begin{enumerate}[(i)]
\item Note that $ W(t) ^* $ can be rewritten as 
\begin{align}
W(t)^* = U(t,0) ^* e^{itH_0} = e^{itH_0} U(0,-t) ^* 
\end{align}
since $ e^{itH_0} $ is unitary. By using Theorem \ref{adjoint}, for all $ \eta \in D(H) $ and $ \xi \in D' $, we have
\begin{align}\label{sch6}
\frac{d}{dt} \left\langle \eta , W(t) ^* \xi \right\rangle = \left\langle -iH \eta , W(t) ^* \xi \right\rangle .
\end{align}
On the other hand, we can see that $ W(t)^* \xi $ is strongly differentiable and the derivative becomes
\begin{align}\label{sch7}
\frac{d}{dt}  W(t) ^* \xi  = U(t,0) ^* (iH_1 (t) ^* + iH_0)e^{itH_0} \xi = i W(t)^* H^* \xi ,
\end{align}
in the same way as \eqref{sch2}. Hence, by \eqref{sch6} and \eqref{sch7}, we have
\begin{align}
\left\langle -iH \eta , W(t) ^* \xi \right\rangle = \left\langle \eta , i W(t) ^* H^* \xi  \right\rangle ,
\end{align}
which implies that $ W(t)^* \xi \in D(H^*) $ and $ iW(t)^* H^*\xi = iH^* W(t)^*\xi $. Therefore, we obtain \eqref{sch8}.
\item Firstly, we show that $ W(t) \xi \in D(B) $ for each $ \xi \in D $. By using \eqref{Bestimate} and Lemma \ref{ess-estimate}, one finds 
\begin{align}\label{BeS}
\| B e^{-itH_0} S_n (t,0)  \xi \| & \le \sum _{j=0} ^n \| B e^{-itH_0}U_j (t,0) \xi \| \no \\
& \le \sum _{j=0} ^\infty C_0 (L_\xi +jb_0 +1) ^{1/2} \frac{|t| ^j }{j!} C^j (L_\xi + (j-1)b +1) ^{1/2} \cdots (L_\xi +1) ^{1/2} \| \xi \| \no \\
& < \infty ,
\end{align}
where the convergence is uniform in $ t $ on any compact set $ K \subset \R $. Since $ W(t) \xi = \sum _{n=0} ^\infty e^{-itH_0}U_n (t,0) \xi $, it follows that $ W(t) \xi \in D(B) $ and that 
\[ B W(t) \xi = \sum _{n=0} ^\infty B e^{-itH_0}U_n (t,0) \xi = \lim _{n\rightarrow \infty }B e^{-itH_0}S_n (t,0) \xi \] 
from the closedness of $ B $. 
Next, we show that $ BW(t)\xi \in D(W(-t)) $. Since $ W(-t) $ is closed, it is sufficient to prove that the sequence 
\[  W(-t) B e^{-itH_0}S_n (t,0) \xi \] 
converges. But, this follows because
\begin{align}
 &\|W(-t) B e^{-itH_0}S_n (t,0) \xi \| \no \\
 &\q= \| e^{itH_0}U (-t,0) B e^{-itH_0}S_n (t,0) \xi \| \no\\
&\q\le \sum _{m,j=0} ^\infty \| e^{itH_0}U_m (-t,0) B e^{-itH_0}U_j (t,0) \xi \| \no\\
& \q \le\sum _{m,j=0} ^\infty \frac{|t| ^m }{m!} C^m (L_\xi + (m+j-1)b + b_0 +1) ^{1/2} \cdots (L_\xi + jb +b_0 +1) ^{1/2}\no \\
& \q \q\times C_0  (L_\xi + jb +1) ^{1/2} \frac{|t| ^j }{j!} C^j (L_\xi + (j-1)b +1) ^{1/2} \cdots (L_\xi +1) ^{1/2} \| \xi \| \no\\
 & \q\le \sum _{N=0} ^\infty \frac{(2|t|) ^N }{N!} C^N (L_\xi + (N-1)b + b_0 +1) ^{1/2} \cdots (L_\xi + b_0 +1) ^{1/2}  C_0  (L_\xi + Nb +1) ^{1/2} \| \xi \| \no\\
&\q < \infty .
\end{align}
Hence, it follows that $ B W(t) \xi \in D(W(-t)) $ and
\begin{align}\label{WBW}
W(-t) B W(t) \xi = \sum _{m,j =0} ^\infty e^{itH_0}U_m (-t,0) B e^{-itH_0} U_j (t,0) \xi , 
\end{align}
where the right hand side converges absolutely. This means that $ D\subset D(W(-t) BW(t)) $.
\item By using \eqref{Bestimate} and Lemma \ref{U-estimate^*}, we get the desired conclusion in the same way as (ii), since $B(t)^*\supset W(-t)^*B^*W(t)^*$ in general.
\item By Assumption \ref{ass4} (II) and \eqref{t'diff}, the derivative of $ B S_n (0,-t) e^{-itH_0} \xi $ becomes
\begin{align}
\frac{d}{dt} B S_n (0,-t) e^{-itH_0} \xi = -i B S_{n-1} (0,-t) e^{-itH_0} H \xi .
\end{align}
From the estimation \eqref{BeS}, one finds that $ \{ (d/dt) B S_n (0,-t) e^{-itH_0} \xi \} _n $ is Cauchy for all $ \xi \in D' $. Hence, the limit $ \lim _{n\to \infty } (d/dt) B S_n (0,-t) e^{-itH_0} \xi = -i \lim _{n\to \infty } B S_{n-1} (0,-t) e^{-itH_0} H \xi $ exists. Due to the fact that $ B $ is closed, this implies $ W(t) H \xi  \in D(B) $ and 
\begin{align}
\frac{d}{dt} B S_n (0,-t) e^{-itH_0} \xi \to -i B W(t) H \xi , \q (n\to \infty ) ,
\end{align}
uniformly on any finite interval $ K $. Since the function $ \frac{d}{dt} B S_n (0,-t) e^{-itH_0} \xi $ is strongly continuous, so is its uniform limit $ -i B W(t) H \xi $. 
Then, by exchanging limit and integration, we have 
\begin{align}
\lim _{n \to \infty } \int _0 ^t d\tau\, \frac{d}{d\tau } B S_n (0,-\tau ) e^{-i\tau H_0} \xi = -i \int _0 ^t d\tau  B W(\tau ) H \xi  ,
\end{align}
where the convergence is uniform on $ K $. Since the left hand side is equal to $ B W(t) \xi - B\xi $, we 
obtain
\begin{align}
B W(t) \xi - B\xi = -i \int _0 ^t d\tau  B W(\tau ) H \xi  ,
\end{align}
which implies that $ BW(t) \xi  $ is strongly continuously differentiable in $t\in K $. Since $ K $ is arbitrary, one concludes that \eqref{BWdiff} holds.
\end{enumerate}
\end{proof}
\begin{proof}[Proof of Theorem \ref{scheq-and-heieq}] By Lemma \ref{heiLem}, for each $ \xi , \eta \in D' $, we know that the functions $ W(-t) ^* \eta $ and $ BW(t) \xi $ are strongly differentiable. Hence we have
\begin{align}
\frac{d}{dt} \left\langle \eta , B(t) \xi \right\rangle &= \frac{d}{dt} \left\langle W(-t) ^* \eta , B W(t) \xi \right\rangle \no\\
& = \left\langle -i W(-t) ^* H^* \eta , B W(t) \xi \right\rangle + \left\langle W(-t)^* \eta , -iBW(t) H \xi \right\rangle \no\\
& = \left\langle (iH)^* \eta , W(-t) B W(t) \xi \right\rangle - \left\langle W(t)^* B^* W(-t) ^* \eta , iH \xi \right\rangle .
\end{align}
Therefore, by \eqref{B*}, we obtain \eqref{wHei2}.
\end{proof}

In the rest of this section, we give a proof of Theorem \ref{extra}. We denote the closure of $ B_0'(t) $ by the same symbol.

\begin{Lem}\label{extra1} Under Assumptions \ref{ass1}, \ref{ass2}, \ref{ass4} and \ref{ass5}, the following (i)-(iv) hold.
\begin{enumerate}[(i)]
\item $ B_0' (t) (A+1) ^{-1/2} $ is bounded and there exists 
a constant $ C_1 \ge 0 $ independent of $ t\in \R $ such that
\[ \| B_0 ' (t) (A+1) ^{-1/2} \| \le C_1. \]

\item For all $ t\in \R $ and $ L\ge 0 $, $ \xi \in V_L $ implies $ B_0 ' (t) \xi \in V_{L+b_0} $, where $ b_0 $ is given in Assumption \ref{ass4} (II).

\item For all $ t\in \R $, $ D\subset D(U(0,t) H_1 (t) B(t) U(t,0)) \cap D(U(0,t) B(t) H_1 (t) U(t,0)) \cap D(U(0,t) B_0 ' (t) U(t,0)) $, and for each $ \xi \in D $,
\begin{align}
U(0,t) H_1 (t) B(t) U(t,0) \xi &= \sum _{m,n =0} ^\infty U_m(0,t) H_1 (t) B(t) U_n(t,0) \xi ,\\
U(0,t) B(t) H_1 (t) U(t,0) \xi &= \sum _{m,n =0} ^\infty U_m(0,t) B(t) H_1 (t) U_n(t,0) \xi ,\\
U(0,t) B_0 ' (t) U(t,0) \xi & = \sum _{m,n =0} ^\infty U_m(0,t) B_0 ' (t) U_n(t,0) \xi ,
\end{align}
where the right hand side in each equations above converges absolutely and uniformly in $ t $ on any compact interval, and does not depend upon the summation order. 

\item For all $ \xi \in D(A^{1/2}) \cap D(H_0) $ and $ t\in \R $, 
\begin{align}
B_0 ' (t) \xi = e^{itH_0} [iH_0 , B] e^{-itH_0} \xi .
\end{align}
\end{enumerate}
\end{Lem}

\begin{proof} The statement (i) follows from Assumption \ref{ass5} (III). 

We prove (ii). Let $ L \ge 0 $ and $ \xi \in V_L $. By the definition of $ B_0' (t) $,
\begin{align}
B_0 ' (t) \xi = \lim _{\delta \to 0} \frac{B_0(t+\delta ) \xi - B_0 (t)\xi }{\delta} , \q t\in \R .
\end{align}
Since $ H_0 $ and $ A $ are strongly commuting, it follows that $ B_0 (t) \xi \in V_{L+b_0} $ for all $ t $. Hence, $ B_0 '(t) \xi \in V_{L+b_0} $ by the closedness of $ V_{L+b_0} $.

By using (i)-(ii), Lemma \ref{rel-bdd} and \ref{rel-bdd^*}, we can check (iii) in the same manner as in the proof of Lemma \ref{heiLem} (ii).

We prove (iv). For each $ \eta \in D(H_0) $ and $ \xi \in D(A^{1/2}) \cap D(H_0) $, we have
\begin{align}\label{extra7}
\frac{d}{dt} \Expect{ \eta , B_0 (t) \xi } = \Expect{ -iH_0 \eta , B_0 (t) \xi } + \Expect{ \eta , B_0 (t) (-iH_0) \xi } .
\end{align}
On the other hand, By Assumption \ref{ass5} and the definition of $ B_0' (t) $, we have
\begin{align}\label{extra8}
\frac{d}{dt} \Expect{ \eta , B_0 (t) \xi } = \Expect{ \eta , B' _0 (t) \xi }.
\end{align}
Comparing \eqref{extra7} and \eqref{extra8}, we obtain $ B_0 (t) \xi \in D(H_0) $ and 
\begin{align}\label{extra9}
B' _0(t) \xi = [iH_0 , B_0(t)]\xi = e^{itH_0} [iH_0 , B] e^{-itH_0} \xi .
\end{align}

\end{proof}

\begin{proof}[Proof of Theorem \ref{extra}] By using Lemma \ref{extra1} (i) and (ii), it follows that for each $ \xi \in D $ and $ m,n =0,1,2,\dots $, the function $ \R \ni t \mapsto S_m (0, t) B_0(t) S_n (t,0) \xi $ is strongly differentiable and the derivative becomes
\begin{align}\label{extra2}
\frac{d}{dt} S_m (0, t) B_0(t) S_n (t,0) \xi =&  S_{m-1} (0, t) i H_1 (t) B_0(t) S_n (t,0) \xi + S_{m} (0, t) B_0(t) (-iH_1(t)) S_{n-1} (t,0) \xi \no \\
& \q + S_m (0, t) B_0' (t) S_n (t,0) \xi ,
\end{align}
where $ S_{-1} (\cdot , \cdot ) := 0 $. From Lemma \ref{extra1} (iii), one finds that
\begin{align}\label{extra3}
\lim _{m,n \to \infty } \frac{d}{dt} S_m (0, t) B_0(t) S_n (t,0) \xi = U(0 ,t) [iH_1 (t), B_0(t)] U(t,0) \xi + U(0,t) B_0'(t) U(t,0)\xi ,
\end{align}
uniformly in $ t $ on any compact set $ K\subset \R $. Since the right hand side of \eqref{extra2} is strongly continuous by Assumption \ref{ass5} (I), so is the left hand side of \eqref{extra3}. Hence, by exchanging limit and integration, we get
\begin{align}\label{extra4}
\lim _{m,n \to \infty } \int _0 ^t d\tau \, \frac{d}{d\tau } S_m (0, \tau ) B_0(\tau ) S_n (\tau ,0) \xi = \int _0 ^t d\tau \,\Big(  U(0 ,\tau ) [iH_1 (\tau ) , B_0(\tau )] U(\tau , 0) \xi + U(0,\tau ) B_0'(\tau ) U(\tau ,0) \xi \Big) ,
\end{align}
where the convergence is uniform on the compact set $ K $. Since the left hand side of \eqref{extra4} is equal to $ B(t) \xi - B\xi $ due to \eqref{WBW}, one concludes that $ B(t)\xi  $ is strongly continuously differentiable in $t\in K $, and the derivative becomes
\begin{align}
\frac{d}{dt} B(t) \xi & = U(0 ,t) [iH_1 (t ) , B_0(t )] U(t , 0) \xi + U(0,t ) B_0'(t) U(t ,0) \xi \no \\
& =  W(-t) [iH_1 , B] W(t) \xi + U(0,t ) B_0'(t) U(t ,0) \xi . 
\end{align}
Therefore, \eqref{extra5} follows from the arbitrariness of $ K $.

It remains to prove \eqref{extra6.5}. Let $ \xi \in D' $. Note that for all $ t\in \R $, we have $ U(t,0) \xi \in D(A^{1/2}) $ from Theorem \ref{App1}, and $ U(t,0)\xi \in D(H_0) $ from Theorem \ref{sch-existence}. From these facts and Lemma \ref{extra1} (iv) and \eqref{extra5}, we obtain
\begin{align}\label{extra24}
\frac{d}{dt} B(t) \xi = W(-t) [iH , B] W(t) \xi .  
\end{align}
From \eqref{WH} in the proof of Theorem \ref{sch-existence}, we have
\begin{align}\label{extra22}
W(-t) BHW(t) \xi = W(-t) B W(t) H \xi = B(t) H \xi .
\end{align}
By using Lemma \ref{heiLem} (i), we have for all $ \eta \in D' $
\begin{align} 
\Expect{ \eta , W(-t) H BW(t) \xi } & = \Expect{ H^* W(-t) ^* \eta , BW(t) \xi } \no \\
& = \Expect{ W(-t) ^* H^* \eta , BW(t) \xi } \no \\
& = \Expect{ \eta , H W(-t) BW(t) \xi } . \no 
\end{align}
Thus 
\begin{align}\label{extra23}
W(-t) HBW(t) \xi = H W(-t) B W(t) \xi = H B(t) \xi ,
\end{align}
because $ D' $ is dense. Hence, \eqref{extra6.5} follows from \eqref{extra24}, \eqref{extra22} and \eqref{extra23}. 
\end{proof}

\section{Application to QED in Lorenz gauge}\label{QED}
In this section, we apply the general theory obtained in the preceding sections to a mathematical model 
of QED, quantized in the Lorenz gauge. As we emphasized in Introduction, our construction of $ U(t,t') $ does not require that $ H $ be self-adjoint, 
 and Theorems \ref{sch-existence}, and
\ref{scheq-and-heieq} are independent of the self-adjointness of $ H $. 
This method is particularly valid for analyzing Lorenz-gauge QED, whose Hamiltonian is not self-adjoint and not even normal. 
We expect that our theory would be applicable to a wider class of mathematical models of quantum systems.
Other possible applications, including to a model with ordinary self-adjoint Hamiltonians and 
a more detailed analysis of Lorenz-gauge QED, are in progress and will be presented in separated papers.

QED describes a system 
in which the quantum radiation field and the quantum Dirac field are minimally interacting. 
It is well known that, in the Coulomb gauge, one can employ a state space constructed by usual Fock spaces, 
which equip a positive definite metric, at the cost of the Lorentz covariance. 
In this formulation, the Hamiltonian $ H $ is self-adjoint \cite{MR2541206}, hence, there clearly 
exists the time evolution operator $ e^{-itH} $ 
such that $ \xi (t)= e^{-itH}\xi $ or $ B(t) = e^{itH} B e^{-itH} $ is the unique solutions 
of the initial value problems \eqref{sch00} or \eqref{hei00}, respectively.
In contrast to the case of Coulomb gauge, in the Lorenz gauge, the Hamiltonian is neither self-adjoint nor 
normal in consequence of the inevitability of an indefinite metric \cite{MR2412280}, 
and hence the time evolution operator $ e^{-itH} $ does not necessarily exist. As a result, 
even the existence of solutions of \eqref{sch00} and \eqref{hei00} becomes a highly nontrivial problem. It does not seem to be easy to apply the general theory of evolution operators 
through hard analyses of the resolvent of the Hamiltonian of Lorenz-gauge QED. But our general theory works well to construct an appropriate time evolution as we will see in the present section.

\subsection{Radiation fields}\label{radiation}
We introduce the photon field quantized in the Lorenz gauge. 

We adopt as the one-photon Hilbert space
\begin{align}
\H _\mathrm{ph} := L^2 (\R ^3 _\kk ; \C ^4) .
\end{align}
The above $ \R ^3 _\kk := \{ \mathbf{k} = (k^1,k^2,k^3) \, | \, k^j \in \R , \, j=1,2,3 \} $ physically represents the momentum space of photons. If there is no danger of confusion, we omit the subscript $ \kk $ in $ \R ^3 _\kk $. $ \H _\rm{ph} $ can be identified as $ \oplus ^4 L^2 (\R ^3 _\kk ) $. We freely use this identification. The Hilbert space for the quantized radiation field in the Lorenz gauge is given by 
\begin{align}
\F _\rm{ph}:= \op _{n=0} ^\infty \ot _\rm{s} ^n \H _\rm{ph} = \Big\{ \Psi = \{ \Psi ^{(n)} \} _{n=0} ^\infty \, \Big| \, \Psi ^{(n)} \in \ot _\rm{s} ^n \H _\rm{ph} , \, \, \left\| \Psi \right\| ^2 _{\F _\rm{ph}} = \sum _{n=0} ^\infty \left\| \Psi ^{(n)} \right\| _{\ot _\rm{s} ^n \H _\rm{ph} } ^2 <\infty \Big\} ,
\end{align}
the Boson Fock space over $ \H _\mathrm{ph} $, where $\textstyle \ot _\rm{s} ^n $ denotes the $ n $-fold symmetric tensor product with the convention $\textstyle \ot _\rm{s} ^0 \H _\rm{ph} := \C $.

One-photon Hamiltonian in $ \H _\rm{ph} $ is the multiplication operator by the function $ \omega (\kk ) := |\kk | \, (\kk \in \R ^3) $.
We also denote by $\omega$ the matrix valued function 
\begin{align}
 \kk\mapsto \begin{pmatrix}
 \omega(\kk) & 0 & 0& 0\\
 0 & \omega(\kk) & 0 & 0\\
 0 & 0 & \omega(\kk) & 0 \\
 0 & 0&0&\omega(\kk)
 \end{pmatrix}
 \end{align}
 and the multiplication operator by it.
 Then, the free Hamiltonian of the quantum radiation field is given by
\begin{align}
H_\mathrm{ph} := \mathrm{d}\Gamma _\bb (\omega ) := \op _{n=0} ^\infty \omega ^{(n)} , 
\end{align}
the second quantization of $ \omega $, where $ \omega ^{(0)} := 0 $ and $ \omega ^{(n)} \, (n\ge 1) $ is defined by
\begin{align}
\omega ^{(n)} := \overline{ \Big( \sum _{j=1} ^n I \otimes \dots \otimes  I \otimes \stackrel{j\text{-th}}{\omega} \otimes I \otimes \dots \otimes I \Big) \upharpoonright \hot ^n D(\omega ) } : \ot _\rm{s} ^n \H _\rm{ph} \to \ot _\rm{s} ^n \H _\rm{ph} ,
\end{align}
and $ \hat{\otimes } $ denotes the algebraic tensor product.

Let $ g = (g_{\mu\nu} ) _{\mu , \nu =0,1,2,3} $ be the $ 4\times 4 $ matrix given by
\begin{align}
g= \begin{pmatrix} 1 &0&0&0 \\ 0&-1&0&0 \\ 0&0&-1&0 \\ 0&0&0&-1 \end{pmatrix} .
\end{align}
Now we introduce the indefinite metric on $ \F _\mathrm{ph} $. The matrix $ g $ naturally defines the unitary operator acting on $ L^2 (\R ^3 _\kk ; \C ^4) $. We denote it by the same symbol $ g $. We define $ \eta $ by the second quantization of $ -g $, i.e.,
\begin{align}
\eta := \Gamma _\bb (-g) := \op _{n=0} ^\infty \ot ^n \, (-g) : \F _\mathrm{ph} \to \F _\mathrm{ph} .
\end{align}
Then $ \eta $ is unitary and satisfies $ \eta ^* = \eta , \, \eta ^2 =I $. By using $ \eta $ we introduce an indefinite metric on $ \F _\mathrm{ph} $ by 
\begin{align}\label{metric}
\left\langle \Psi | \Phi  \right\rangle := \left\langle \Psi , \eta \Phi  \right\rangle _{\F _\rm{ph}}, \q \Psi , \Phi \in \F _\rm{ph} .
\end{align}
In order to define the adjoint with respect to indefinite metric \eqref{metric}, we introduce the $ \eta $-adjoint. For a densely defined linear operator $ T $ on $ \F _\rm{ph} $, the adjoint operator $ T^\dagger $ with respect to the metric $ \left\langle \cdot | \cdot \right\rangle $ is defined by 
\begin{align}
T^\dagger := \eta T^* \eta .
\end{align}
It follows that
\begin{align}
\Expect{ \Psi | T\Phi } = \Expect{ T ^\dagger \Psi | \Phi } , \q \Psi \in D(T^\dagger ) , \q \Phi \in D(T) .
\end{align}

We introduce notions of $ \eta $-symmetry, $ \eta $-self-adjointness and $ \eta $-unitarity \cite{MR2533876,MR2412280} below.

\begin{Def}\label{eta-sa}
\begin{enumerate}[(i)] 
\item A densely defined linear operator $ T $ is $ \eta $-symmetric if $ T \subset T^\dagger $.

\item A densely defined linear operator $ T $ is $ \eta $-self-adjoint if $ T^\dagger =T $.

\item A densely defined linear operator $ T $ is essentially $ \eta $-self-adjoint if $ \overline{T} $ is $ \eta $-self-adjoint.

\item A densely defined linear operator $ T $ is $ \eta $-unitary if $ T $ is injective and $ T ^\dagger = T^{-1} $.
\end{enumerate}
\end{Def}

\begin{Lem}\label{eta-lem} 
\begin{enumerate}[(i)] 
\item $ T $ is $ \eta $-symmetric if and only if $ \eta T $ is symmetric.

\item $ T $ is $ \eta $-self-adjoint if and only if $ \eta T $ is self-adjoint. 

\item $ T $ is essentially $ \eta $-self-adjoint if and only if $ \eta T $ is essentially self-adjoint.

\item If $ T $ is $ \eta $-symmetric then $ T $ is closable.

\item Let $ T $ be $ \eta $-self-adjoint and $ \eta T $ is essentially self-adjoint on a subspace $ D $, Then $ D $ is a core of $ T $.
\end{enumerate}
\end{Lem}

\begin{proof} See \cite{MR2533876}.
\end{proof}

Note that the free Hamiltonian $ H_\rm{ph} $ is self-adjoint and $ \eta $-self-adjoint.

The annihilation operator $ a(F) $ with $ F \in \H _\mathrm{ph} $ is defined to be a densely defined closed linear operator on $ \F _\mathrm{ph} $ whose adjoint is given by
\begin{align}
(a(F) ^* \Psi ) ^{(0)} = 0 , \q (a(F) ^* \Psi ) ^{(n)} = \sqrt{n} S_n (F \otimes \Psi ^{(n-1)}) , \q n \ge 1 , \q \Psi \in D(a(F)^*), 
\end{align}
where $ S_n $ denotes the symmetrization operator on $ \otimes ^n \H _\rm{ph} $, i.e. $ S_n (\otimes ^n \H _\rm{ph}) = \otimes _\rm{s} ^n \H _\rm{ph} $. We note that $ a(F) $ is anti-linear in $ F $ and $ a(F')^* $ linear in $ F' $. As is well known, the creation and annihilation operators leave the finite particle subspace 
\begin{align}
\F _{\bb , 0} (\H _\rm{ph}) := \Big\{ \{ \Psi ^{(n)} \} _{n=0} ^\infty \in \F _\rm{ph} \, \Big| \,  \Psi ^{(n)} =0 \, \text{for all $ n\ge n_0 $ with some $n_0$} \Big\} 
\end{align}
invariant and satisfy the canonical commutation relations:
\begin{align}
[a(F) , a (F' ) ^* ] = \left\langle F, F' \right\rangle _{\H _\rm{ph}} , \q [a(F) , a(F' ) ] = [a (F) ^* , a(F' )^* ]=0,
\end{align}
on $ \F _{\bb, 0} (\H _\rm{ph}) $. For each $ f \in L^2 (\R _\kk ^3) $, we use the notation:
\begin{align}
& a^{(0)} (f) := a (f,0,0,0) , \q a^{(1)} (f) := a (0,f,0,0) , \\
& a^{(2)} (f) := a (0,0,f,0) , \q a^{(3)} (f) := a (0,0,0,f).
\end{align}
Then, the operator equalities
\begin{align}
& a^{(0)} (f) ^\dagger = -a^{(0)} (f) ^* , \\ 
& a^{(j)} (f) ^\dagger = a^{(j)} (f) ^* , \q j=1,2,3,
\end{align}
hold. For each $ f\in L^2 (\R _\kk ^3) $ and $ \mu =0,1,2,3 $, we define
\begin{align}
& a_\mu (f) := a \big( f e _\mu ^{(0)} ,\,  f e _\mu ^{(1)} ,\, f e _\mu ^{(2)} ,\, f e _\mu ^{(3)} \big) , 
\end{align}
where $ e^{(\lambda  )} (\kk ) = (e^{(\lambda )} _\mu (\kk ) ) _{\mu =0} ^3 \in \C ^4 \, (\lambda = 0,1,2,3) $ are the polarization vectors satisfying 
\begin{align}
\sum _{ \lambda , \lambda ' =0} ^3 e ^{ ( \lambda ) } _\mu (\kk ) g_{\lambda \lambda '}  e _\nu ^{(\lambda ' ) } (\kk ) ^* &= g_{\mu \nu } , \q \text{a.e.} \;  \kk \in \R ^3 , \q \mu , \nu =0,1,2,3.
\end{align}
In this paper, for simplicity, we choose the special Lorentz frame, that is, 
\begin{align}
& e^{(0)} (\kk ) := (1,0,0,0) , \q e^{(3)} := (0, \frac{\kk }{|\kk |}) , \\
& e^{(r )} (\kk ) := (0, \mathbf{e} ^{(r ) } (\kk) ) , \q r =1,2, 
\end{align}
where $ \mathbf{e} ^{(r )} \, (r =1,2) $ are $ \R ^3 $-valued continuous functions on the nonsimply connected space $ K_0 := \R ^3 \backslash \{ (0,0,k^3) \, | \, k^3 \in \R \} $ such that, for all $ \kk \in K_0 $,
\begin{align}
\mathbf{e} ^{(r )} (\kk ) \cdot \mathbf{e} ^{(r ' )} (\kk ) = \delta _{rr'} , \q \mathbf{e} ^{(r )} (\kk ) \cdot \kk =0, \q r,r' =1,2.
\end{align}
We often use the notation $ a_\mu ^\dagger (f) = a_\mu (f) ^\dagger $. Then, $ a_\mu (f) $ and $ a_\mu ^\dagger (f) $ are closed, and satisfy the commutation relations:
\begin{align*}
[a_\mu (f) , a_\nu ^\dagger (g) ] & = -g _{\mu\nu } \left\langle f,g \right\rangle _{L^2 (\R ^3)} , \\
[a_\mu (f) , a_\nu (g)] & = [a_\mu ^\dagger (f) , a_\nu ^\dagger (g) ] =0 ,
\end{align*}
on $ \F _{\bb , 0} (\H _\rm{ph}) $.

For all $ f \in L^2 (\R ^3 _\xx ) $ satisfying $ \hat{f}/\sqrt{\omega} \in L^2 (\R ^3 _\kk ) $, we set 
\begin{align}
A_\mu (0,f) := a_\mu \Big( \frac{\hat{f^*}}{\sqrt{2\omega }} \Big) + a_\mu ^\dagger \Big( \frac{\hat{f}}{\sqrt{2\omega }} \Big) ,
\end{align}
where $ \hat{f} $ denotes the Fourier transform of $ f $, and $ f^* $ denotes the complex conjugate of $ f $. The functional $ \S (\R ^3 _\xx ) \ni f \mapsto A_\mu (0,f) $ gives an operator-valued distribution (Cf. \cite{AraiFock} Definition 7-1)
 acting on $ (\F _\rm{ph} , \F _{\bb , 0} (\H _\rm{ph})) $ and it is called the quantized radiation field at time $ t=0 $.
Now, fix $ \chi _\rm{ph} \in L^2 (\R ^3_\xx ) $ such that it is real and satisfies $ \hat{\chi _\rm{ph}} / \sqrt{\omega} \in L^2 (\R ^3 _\kk ) $. We set
\begin{align}
& A_\mu (\xx ) := A_\mu (0, \chi _\rm{ph} ^\xx ) , \\
& \chi _\rm{ph} ^\xx (\yy ) := \chi _\rm{ph} (\yy - \xx ) , \q \yy \in \R ^3 .
\end{align}
$ A_\mu (\xx ) $ is called the point-like quantized radiation field with momentum cutoff $ \hat{\chi _\mathrm{ph}} $.
As will be seen later, for real-valued $ f $, the closures of $ A_\mu (f) , \, \mu =0,1,2,3, $ are $ \eta $-self-adjoint but not even normal.

\subsection{Dirac fields}

Next, we define the quantized Dirac field. We adopt as the one-electron Hilbert space 
\begin{align}
\H _\mathrm{el} := L^2 (\R ^3 _\pp ; \C ^4) ,
\end{align}
where $ \R ^3 _\pp := \{ \mathbf{p} = (p^1,p^2,p^3) \, | \, p^j \in \R , \, j=1,2,3 \} $ physically represents the momentum space of electrons. The Hilbert space for the quantized Dirac field is given by 
\begin{align}
\F _\mathrm{el} := \op _{n=0} ^\infty \wg ^n \H _\rm{el} = \Big\{ \Psi = \{ \Psi ^{(n)} \} _{n=0} ^\infty \, \Big| \, \Psi ^{(n)} \in \wg ^n \H _\rm{el} , \, \, \left\| \Psi  \right\| ^2 _{\F _\rm{el}} = \sum _{n=0} ^\infty \left\| \Psi ^{(n)} \right\| _{\wg ^n \H _\rm{el} } ^2 <\infty \Big\} , 
\end{align}
the Fermion Fock space over $ \H _\mathrm{el} $, where $ \wedge ^n $ denotes the $ n $-fold anti-symmetric tensor product with the convention $ \wedge ^0 \H _\rm{el} := \C $.

We denote the mass of the Dirac particle by $ M>0 $. One-electron Hamiltonian in $ \H _\rm{el} $ is the multiplication operator by the function $ \textstyle E_M (\pp ) := \sqrt{\pp ^2 + M^2} \, \, (\pp \in \R ^3) $. The Hamiltonian of the free quantum Dirac field is given by
\begin{align}
H_\mathrm{el} := \mathrm{d}\Gamma _\ff (E_M ) := \op _{n=0} ^\infty E_M ^{(n)} , 
\end{align}
where
\begin{align}
E_M ^{(n)} := \overline{ \Big( \sum _{j=1} ^n I \otimes \dots \otimes  I \otimes \stackrel{j\text{-th}}{E_M} \otimes I \otimes \dots \otimes I \Big) \upharpoonright \hot ^n D(E_M ) } : \wg ^n \H _\rm{el} \to \wg ^n \H _\rm{el} .
\end{align}
The operator $ H_\rm{el} $ is self-adjoint and non-negative.

Let $ \gamma ^\mu \, (\mu =0,1,2,3) $ be $ 4\times 4 $ gamma matrices, i.e., $ \gamma ^0 $ is hermitian and $ \gamma ^j \, (j=1,2,3) $ are anti-hermitian, satisfying
\begin{align}
\{ \gamma ^\mu , \gamma ^\nu \} =2g^{\mu\nu}, \q \mu, \nu =0,1,2,3,
\end{align}
where $ \{ X,Y \} := XY+YX $. Let $ \alpha ^\mu := \gamma ^0 \gamma ^\mu , \, \beta := \gamma ^0 $, and let $ s_1 := \frac{i}{2} \gamma ^2 \gamma ^3 , \, s_2 := \frac{i}{2} \gamma ^3 \gamma ^1 , \, s_3 := \frac{i}{2} \gamma ^1 \gamma ^2 $. Let $ u_s (\pp ) = (u_s ^l (\pp )) _{l=1} ^4 \in \C ^4 $ describe the positive energy part with spin $ s = \pm 1/2 $ and $ v_s (\pp ) = (v_s ^l  (\pp )) _{l=1} ^4 \in \C ^4 $ the negative energy part with spin $ s $, that is,
\begin{align}
& ( { \boldsymbol \alpha } \cdot \pp + \beta M) u_s (\pp ) = E_M (\pp ) u_s (\pp ) , \q ( \mathbf{s} \cdot \pp ) u_s (\pp ) = s|\pp | u_s (\pp ) , \\
& ( { \boldsymbol \alpha } \cdot \pp + \beta M) v_s (\pp ) = -E_M (\pp ) v_s (\pp ) , \q ( \mathbf{s} \cdot \pp ) v_s (\pp ) = s|\pp | v_s (\pp ) , \q \pp \in \R ^3 .
\end{align}
These form an orthogonal base of $ \C ^4 $,
\begin{align}
u_s (\pp ) ^* u_{s'} (\pp ) = v_s (\pp ) ^* v_{s'} (\pp ) = \delta _{ss'} E_M (\pp ) , \q u_s (\pp ) ^* v_{s'} (\pp ) = 0 ,
\end{align}
and satisfy the completeness,
\begin{align*}
\sum _s \big( u^l_s (\pp ) u^{l' } _s (\pp ) ^* + v^l _s(\pp ) v ^{l' } _s(\pp ) ^*  \big) = 2 \delta _{ll' } E_M (\pp ) .
\end{align*}

The annihilation operator $ B(G) $ with $ G \in \H _\mathrm{el} $ is defined to be a \textit{bounded} operator on $ \F _\mathrm{el} $ whose adjoint is given by
\begin{align}
(B(G) ^* \Psi ) ^{(0)} = 0, \q (B(G) ^* \Psi ) ^{(n)} = \sqrt{n} A_n (G \otimes \Psi ^{(n-1)}) , \q n \ge 1 ,  \q \Psi = \{ \Psi ^{(n)} \} _{n=0} ^\infty \in \F _\rm{el} ,
\end{align}
where $ A_n $ denotes the anti-symmetrization operator on $ \otimes ^n \H _\rm{el} $, i.e. $ A_n (\otimes ^n \H _\rm{el}) = \wedge ^n \H _\rm{el} $. $ B(G) $ is anti-linear in $ G $ and $ B(G')^* $ linear in $ G' $. It is well known that the operator norm of $ B(G) $ and $ B(G') ^* $ is given by \cite{AraiFock}
\begin{align}\label{fermi-est}
\| B(G) \| = \| G \| _{\H _\rm{el}}  , \q \| B(G')^* \| = \| G' \| _{\H _\rm{el}} .
\end{align}
For each $ g \in L^2 (\R _\pp ^3) $, we use the notation 
\begin{align*}
& b _{1/2} (g) := B(g,0,0,0) , && b_{-1/2} (g) := B(0,g,0,0) , \\
& d _{1/2} (g) := B(0,0,g,0) , && d_{-1/2} (g) := B(0,0,0,g) ,
\end{align*}
and $ b_s^* (g) := b_s (g)^* $, $ d_s ^* (g) := d_s (g)^* , \, (s= \pm 1/2) $. Then, we have the canonical anti-commutation relations: 
\begin{align}
& \{ b_s (g) , b_{s'}^* (g') \} =\{ d_s (g) , d_{s'} ^* (g') \} = \delta _{ss'} \left\langle g , g' \right\rangle _{L^2 (\R ^3 _\pp )} , \\
& \{ b_s (g) , b_{s'} (g') \} =\{ d_s (g) , d_{s'} (g') \} = \{ b_s (g) , d_{s'} (g') \} =\{ b_s (g) , d_{s'} ^* (g') \} =0 .
\end{align}

For all $ g \in L^2 (\R ^3 _\xx ) $, we set 
\begin{align}
\psi _l (0, g) := & \sum _{s= \pm 1/2} \Bigg( b_s \Bigg( \frac{\hat{g^*} \cdot u^l _s {}^* }{\sqrt{2E_M}} \Bigg) + d^*_s \Bigg( \frac{\hat{g} \cdot  \widetilde{v} ^l_s }{\sqrt{2E_M}} \Bigg) \Bigg) ,
\end{align}
where $ \widetilde{v} _s ^l (\pp ) := v_s ^l (-\pp ) $. The functional $ \S (\R ^3 _\xx ) \ni g \mapsto \psi _l (0,g) $ is called the quantized Dirac field at time $ t=0 $. 
Now, fix $ \chi _\rm{el} \in L^2 (\R ^3 _\xx ) $. We set
\begin{align}
& \psi _l (\xx ) := \psi_l (0, \chi _\rm{el} ^\xx ) , \\
& \chi _\rm{el} ^\xx (\yy ) := \chi _\rm{el} (\yy - \xx ) , \q \yy \in \R ^3 .
\end{align}
$ \psi _l (\xx ) $ is called the point-like quantized Dirac field with momentum cutoff $ \hat{\chi _\mathrm{el}} $.
For each $ \xx \in \R ^3 $ and $ \mu =0,1,2,3 $, we define the current operator $ j^\mu (\xx ) $ by
\begin{align}
j^\mu (\xx ) := \sum _{l,l' =1} ^4  \psi  _l (\xx )^* \alpha ^\mu _{ll'} \psi _{l'} (\xx ) .   
\end{align}
Then $ j^\mu (\xx ) $ is bounded and self-adjoint.

\subsection{Total Hamiltonian}

The Hilbert space of state vectors for QED in Lorenz gauge is taken to be 
\begin{align}
\F _{\mathrm{tot}} := \F _\rm{el} \ot \F _\rm{ph} .
\end{align}
This Hilbert space can be identified as 
\begin{align}
\mathcal{F}_\mathrm{tot}=\op_{n=0}^\infty \left( \F _\rm{el} \ot\left( \ot_{s}^n \mathcal{H}_\mathrm{ph}\right) \right).
\end{align}
We freely use this identification.
The free Hamiltonian is 
\begin{align}
H_\rm{fr} := H_\rm{el} \ot I + I \ot H_\rm{ph} ,
\end{align}
where the subscript $ \rm{fr} $ in $ H_\rm{fr} $ means \textit{free}.

We introduce an indefinite metric on $ \F _{\mathrm{tot}} $ by
\begin{align}
\left\langle \Psi | \Phi \right\rangle := \left\langle \Psi , I \otimes \eta \Phi  \right\rangle _{\F _\mathrm{tot}}, \q \Psi , \Phi \in \F _{\mathrm{tot}} .
\end{align}
Then, $ \eta $-adjointness, $ \eta $-symmetricity, $ \eta $-self-adjointness and $ \eta $-unitarity are defined on $ \F _\rm{tot} $, in the same way as in subsection \ref{radiation}, by replacing $ \eta $ with $ I \otimes \eta $.

We introduce the minimal interaction between the quantized Dirac field and the quantized radiation field. We denote the charge of the Dirac particle by $ e \in \Real $. Let  $ N_\bb := \mathrm{d}\Gamma _\bb (I _{H_\rm{ph}}) $ be
the photon number operator and $ \chi _\rm{sp} \in L^1 (\R ^3) $ be a real-valued function on $ \R ^3 $ playing the role of a spacial cut-off.
Our interaction Hamiltonian $ H_\rm{int} $ is defined as 
\begin{align}
& D(H_\rm{int}) = D(I \otimes N_\bb ^\fr{1}{2}) , \\
& H_\rm{int} \Psi = e \int _{\R ^3} d\xx \, \chi _{\mathrm{sp}} (\xx ) j^\mu (\xx ) \ot A_\mu (\xx ) \Psi , \q \Psi \in D(H_\rm{int}) ,
\end{align}
where the integral on the right hand side is taken in the sense of strong Bochner integral, and we used the standard Einstein notation in which 
the summation over repeated indices with one upper and the other lower is understood.  As will be seen in later, the operator $ H_\rm{int} $ is well-defined since 
\begin{align}
\int _{\R ^3} d\xx \, | \chi _\rm{sp} (\xx )| \,  \| j^\mu (\xx ) \otimes A_\mu (\xx ) \Psi  \|  <\infty , \q \Psi \in D(I \otimes N_\bb ^\fr{1}{2}) .
\end{align}
Moreover, $ H_\rm{int} $ is essentially $ \eta $-self-adjoint.

The quantum system under consideration is described by the Hamiltonian
\begin{align}
H _\rm{QED} := H_\mathrm{fr} + \overline{H_\mathrm{int}}  \q \text{in} \q  \F _\rm{tot} .
\end{align}
The time evolution of the quantum fields $ A_\mu , \psi _l $ is generated by the Heisenberg equations:
\begin{align}
& \frac{d}{dt} A^\mu (t,f) = [iH_\rm{QED}  , A^\mu (t,f)] , \\
& \frac{d}{dt} \psi _l (t,g) = [iH_\rm{QED}  , \psi _l (t,g)] .
\end{align}
It is easy to find \textit{formal} solutions of these equations:
\begin{align}
& A^\mu (t,f ) = e^{itH_\rm{QED} } A^\mu (0,f) e^{-itH_\rm{QED} }, \\
& \psi _l (t, g ) = e^{itH_\rm{QED} } \psi _l (0,g) e^{-itH_\rm{QED} } .
\end{align}
However, this does not immediately make sense because our QED Hamiltonian in Lorenz gauge is neither self-adjoint nor normal, and therefore we cannot define the time evolution operational $ e^{-itH_\rm{QED} } $ through the operational calculus.

\subsection{$\eta $-self-adjointness}

In this subsection, we will prove the $ \eta $-self-adjointness of $ A_\mu (0,f) , \, H_\rm{int} $, and $ H_\rm{QED} $ under some suitable conditions.

\begin{Lem}\label{a-est}
\begin{enumerate}[(i)]
\item For all $ f\in L^2 (\R ^3 _\kk ) $ and $ \Psi \in D(N_\bb ^{1/2}) $, 
\begin{align}
\left\| a_\mu (f) \Psi \right\| _{\F _\rm{ph}} & \le \left\| f \right\| _{L^2 (\R ^3 _\kk )} \, \left\| N_\bb ^{1/2} \Psi \right\| _{\F _\rm{ph}} , \\
\left\| a_\mu ^\dagger (f) \Psi \right\| _{\F _\rm{ph}} & \le \left\| f \right\| _{L^2 (\R ^3 _\kk )} \, \left\| (N_\bb +1)^{1/2} \Psi \right\| _{\F _\rm{ph}} .
\end{align}

\item For all $ f\in D(\omega ^{-1/2}) $ and $ \Psi \in D(H_\rm{ph} ^{1/2}) $, 
\begin{align}
\left\| a_\mu (f) \Psi \right\| _{\F _\rm{ph}} & \le \left\| f / \sqrt{\omega } \right\| _{L^2 (\R ^3 _\kk )}  \, \left\| H_\rm{ph} ^{1/2} \Psi \right\| _{\F _\rm{ph}}, \\
\left\| a_\mu ^\dagger (f) \Psi \right\| _{\F _\rm{ph}} & \le \left\| f / \sqrt{\omega } \right\| _{L^2 (\R ^3 _\kk )}  \, \left\| H_\rm{ph} ^{1/2} \Psi \right\| _{\F _\rm{ph}} + \left\| f \right\| _{L^2 (\R ^3 _\kk )} \, \left\| \Psi \right\| _{\F _\rm{ph}}.
\end{align}
\end{enumerate}
\end{Lem}

\begin{proof} The estimates (i) and (ii) are easily proved by applying well known estimations (see \cite{AraiFock} Proposition 4-24, \cite{MR2097788} Section 13.3), and we omit the proof.
\end{proof}

Note that the spectrum of the photon number operator $ N_\bb $ is a purely discrete set $ \{ 0, 1,2,\dots \} $, and that for all integer $ N\ge 0 $, 
\begin{align}
R (E_{N_\bb} ([0, N])) = \op _{n=0} ^N \ot _\rm{s} ^n \H _\rm{ph} \subset \F _\rm{ph} .
\end{align}
For each $ \Psi = \{ \Psi ^{(n)} \} _{n=0} ^\infty \in \F _{\bb ,0} (\H _\rm{ph}) $, we denote by $ N_\Psi $ the maximum photon number of $ \Psi $, that is,
\begin{align}
N_\Psi := \max \{ n\ge 0 \, \big| \, \Psi ^{(n)} \neq 0 \} < \infty.
\end{align}

\begin{Lem}\label{A-eta} Let $ f  $ be a real-valued function satisfying $ f \in L^2 (\R ^3 _\xx ) $ and $ \hat{f} / \sqrt{\omega } \in L^2 (\R_\kk ^3) $. Then, the quantized radiation field $ A_\mu (0, f) $ is essentially $ \eta $-self-adjoint, i.e., $ \overline{A_\mu (0,f)} = \Big( \overline{A_\mu (0,f)} \Big) ^\dagger $.
\end{Lem}

\begin{proof} By Lemma \ref{eta-lem} (iii), it is sufficient to prove that $ \eta A_\mu (0,f) $ is essentially self-adjoint. 

Put $ f_- := \hat{f}/ \sqrt{2\omega } $. Since $ \eta ^* = \eta $, we have
\begin{align}
(\eta A_\mu (0,f)) ^* & = A_\mu (0,f) ^* \eta \no \\
& \supset \big( a_\mu ( f_- ) ^*  + a_\mu ^\dagger ( f_- ) ^* \big) \eta \no \\
& = \eta \eta \big( a_\mu ( f_- ) ^*  + \big( \eta a_\mu ^* ( f_- ) \eta \big) ^* \big) \eta \no \\
& = \eta \big( a_\mu ^\dagger ( f_- )  + a_\mu ( f_- ) \big) \no \\ 
& = \eta A_\mu (0,f) ,
\end{align}
which means that $ A_\mu (0,f) $ is $ \eta $-symmetric. 

We prove the $ \eta $-self-adjointness by Nelson's analytic vector theorem (\cite{MR0493420}, Theorem X.39 and
its corollaries). Clearly, $ A_\mu (0,f) $ and $ \eta $ leaves $ \F _{\bb ,0} (\H _\rm{ph}) $ invariant. By Lemma \ref{a-est} (i), we have for all $ \Psi \in D(N_\bb ^{1/2}) $,
\begin{align}
\| A_\mu (0,f) \Psi \| \le 2 \| f_- \| \, \| (N_\bb +1)^{1/2} \Psi \| .
\end{align}
By the fact that $ \eta $ and $ N_\bb  $ are strongly commuting, one finds for each $ \Psi \in \F _{\bb ,0} (\H _\rm{ph}) $ and $ n=1,2, \dots $,
\begin{align}
\| ( \eta A_\mu (0,f)) ^n \Psi \| \le (N_\Psi + n) ^{1/2} \dots (N_\Psi + 1) ^{1/2} (2\| f_- \| )^n \| \Psi \| .
\end{align}
Thus, one obtains for all $ t>0 $,
\begin{align}
\sum _{n=0} ^\infty \frac{t^n}{n!} \| (\eta A_\mu (0,f)) ^n \Psi \| & \le \sum _{n=0} ^\infty \frac{t^n}{n!} (N_\Psi + n) ^{1/2} \dots (N_\Psi + 1) ^{1/2} (2\| f_- \| )^n \| \Psi \| \no\\
& < \infty ,
\end{align}
by using d'Alembert's ratio test. Therefore, $ \F _{\bb ,0} (\H _\rm{ph}) $ is a $ \eta A_\mu (0,f) $-invariant analytic vector space, and we have the essentially self-adjointness of $ \eta A_\mu (0,f) $ by Nelson's analytic vector theorem.
By Lemma \ref{eta-lem} (v), the assertion follows.
\end{proof}

We denote the closure of $ A_\mu (0, f) $ by the same symbol.

\begin{Lem}\label{int-est} Let $ \Psi \in D((I \otimes N_\bb ) ^{1/2}) $. Then, the followings hold.
\begin{enumerate}[(i)]
\item 
\begin{align}\label{jA}
\| j^\mu (\xx ) \otimes A_\mu (\xx ) \Psi \| \le M_\rm{el} M_\rm{ph} \| (I \otimes N_\bb +1)^{1/2} \Psi \|  ,
\end{align}
where,
\begin{align}
M_\rm{el} := \sum_{\mu=0}^3\sup _{\xx \in\Real^3} \| j^\mu (\xx ) \| , \q M_\rm{ph} := 2 \left\| \frac{ \hat{\chi _\rm{ph}}}{\sqrt{2\omega }} \right\| .
\end{align}

\item The vector-valued function $ \xx \mapsto j^\mu (\xx ) \otimes A_\mu (\xx ) \Psi $ is strongly continuous.

\item Let $ \chi _\rm{sp} $ be a real-valued function satisfying $ \chi _\rm{sp} \in L^1 (\R ^3 _\xx ) $. Then,
\begin{align}\label{spjA}
\int _{\R ^3} d\xx \, |\chi _\rm{sp} (\xx ) | \, \| j^\mu (\xx ) \otimes A_\mu (\xx ) \Psi \| \le \| \chi _\rm{sp} \| _{L^1 (\R ^3)} M_\rm{el} M_\rm{ph} \| (I \otimes N_\bb +1)^{1/2} \Psi \|   <\infty .
\end{align}
\end{enumerate}
\end{Lem}

\begin{proof} First of all, notice that the operator identity $ (I \otimes N_\bb ) ^{1/2} = I\otimes N_\bb ^{1/2} $ follows from operational calculus.

\begin{enumerate}[(i)] \item Note that the current operator $j^\mu(\xx)$ is uniformly bounded in $ \xx \in \R ^3 $, and thus $ M_\rm{el} <\infty $. By using Lemma \ref{a-est} (i), one finds for all $ \Psi \in \F _\rm{el} \hat{\otimes} D(N_\bb ^{1/2}) $,
\begin{align}
\| j^\mu (\xx ) \otimes A_\mu (\xx ) \Psi  \| & = \| \big( j^\mu (\xx ) \ot I \big) \big( I \otimes A_\mu (\xx )\big) \Psi  \|\no \\
& \le M_\rm{el} M_\rm{ph} \| (I \otimes N_\bb +1 )^{1/2} \Psi \| .
\end{align}
Since $ \F _\rm{el} \hat{\otimes} D(N_\bb ^{1/2}) $ is a core of $ (I \otimes N_\bb ) ^{1/2} $, we obtain \eqref{jA}.

\item Let $ \mathbf{P} := (P^1, P^2 , P^3) $, $ P^j := \overline{ \mathrm{d}\Gamma _\ff (p^j) \otimes I + I \otimes \mathrm{d}\Gamma _\bb (k^j) } $, where $ p^j $ and $ k^j $ are the multiplication operators in $ \H _\rm{el} $ and $ \H _\rm{ph} $ respectively. Then, $ P^j , \, j=1,2,3 $ are self-adjoint. Note that the operator $ j^\mu (\xx ) \otimes A_\mu (\xx ) $ can be rewritten as $ e^{-i\mathbf{P} \cdot \xx } j^\mu (\mathbf{0} ) \otimes A_\mu (\mathbf{0}) e^{i\mathbf{P} \cdot \xx } $. Since $ P^j $ and $ I \otimes N_\bb  $ are strongly commuting, we have for all $ \Psi \in D((I \otimes N_\bb ) ^{1/2}) $,
\begin{align}
j^\mu (\xx ) \otimes A_\mu (\xx ) \Psi = e^{-i\mathbf{P} \cdot \xx } j^\mu (\mathbf{0} ) \otimes A_\mu (\mathbf{0}) (I \otimes N_\bb ) ^{-1/2} e^{i \mathbf{P} \cdot \xx } (I \otimes N_\bb ) ^{1/2} \Psi ,
\end{align}
and the right hand side is strongly continuous since $j^\mu (\mathbf{0} ) \otimes A_\mu (\mathbf{0}) (I \otimes N_\bb ) ^{-1/2}$ is a bounded operator by \eqref{jA}.
 
\item The inequality \eqref{spjA} immediately follows from \eqref{jA}. \end{enumerate}\end{proof}

Set
\begin{align}
\F _{\bb , 0} := \hop _{n=0} ^\infty \big( \F _\mathrm{el} \ot (\ot _{s} ^n \H _\mathrm{ph} ) \big) 
\end{align}
where $ \hop _{n=0} ^\infty $ denotes the algebraic direct sum. Let $ N_\Psi $ denote the maximum number of photons of $ \Psi \in \F _{\bb,0} $, namely, for $ \Psi = \{ \Psi ^{(N)} \} _{N=0} ^\infty \in \op _{n=0} ^\infty ( \F _\mathrm{el} \otimes (\ot _\rm{s} ^n H_\rm{ph})) $, $ N_\Psi $ is the largest natural number $ N $ satisfying $ \Psi ^{(N)} \neq 0 $. Define
\begin{align}
\F _N := \op _{n=0} ^N \big( \F _\mathrm{el} \ot (\ot _{s} ^n \H _\mathrm{ph} ) \big) , \q N =0,1,2, \dots .
\end{align}
Then, for all integer $ N\ge 0 $, one finds 
\begin{align}
R (E_{I \otimes N_\bb } ([0,N]) = \F _N .
\end{align}
Since the photon field $A_\mu(\xx)$ creates at most one photon, it follows that if $ \Psi \in \F _{\bb ,0} $ belongs to $ \F _N $, then $ A_\mu(\xx) \Psi \in \F _{N+1} $ for all $\xx \in \R^3$. Thus, since $\F_{N+1}$ is a closed subspace,
we find that $H_\rm{int}\F_{N}$ is contained in $\F_{N+1}$.

\begin{Lem}\label{int-eta} The interaction Hamiltonian $ H_\rm{int} $ is essentially $ \eta $-self-adjoint.
\end{Lem}

\begin{proof} Firstly, we show the $ \eta $-symmetricity of $ H_\rm{int} $. By direct calculation, one finds for each $ \Psi , \Phi \in \F _\rm{el} \hat{\otimes } D(N_\bb ^{1/2}) $, 
\begin{align}\label{etaint-eq}
\Expect{ \Psi , \eta H_\rm{int} \Phi } =  \Expect{ \eta H_\rm{int} \Psi , \Phi } .
\end{align}
Since $ \F _\rm{el} \hat{\otimes } D(N_\bb ^{1/2}) $ is a core of $ (I \otimes N_\bb )^{1/2} $, and since $ H_\rm{int} $ is $ (I \otimes N_\bb )^{1/2} $- bounded, \eqref{etaint-eq} holds for all $ \Psi , \Phi \in D((I \otimes N_\bb )^{1/2}) $, thus $ H_\rm{int} $ is $ \eta $-symmetric.

Next, we show the $ \eta $-self-adjointness by Nelson's analytic vector theorem, similarly to the proof of Lemma \ref{A-eta}. Note that $ H_\rm{int} $ and $ I \otimes \eta $ leaves $ \F _{\bb ,0} $ invariant. Put $ M_\rm{int} :=|e|\, \| \chi _\rm{sp} \| _{L^1 (\R ^3 _\xx )} M_\rm{el} M_\rm{ph} $. By Lemma \ref{int-est} (iii), and by the fact that $ I\otimes \eta $ and $ I \otimes N_\bb  $ are strongly commuting, one finds for each $ \Psi \in \F _{\bb ,0} $ and $ n=1,2, \dots $, 
\begin{align}
\| ( \eta H_\rm{int} )^n \Psi \| \le (N_\Psi + n) ^{1/2} \dots (N_\Psi + 1) ^{1/2} M_\rm{int} ^n \| \Psi \| .
\end{align}
Thus, one obtains for all $ t>0 $,
\begin{align}
\sum _{n=0} ^\infty \frac{t^n}{n!} \| (\eta H_\rm{int}) ^n \Psi \| & \le \sum _{n=0} ^\infty \frac{t^n}{n!} (N_\Psi + n) ^{1/2} \dots (N_\Psi + 1) ^{1/2} M_\rm{int} ^n \| \Psi \| \no \\
& < \infty ,
\end{align}
by using d'Alembert's ratio test. Therefore, $ \F _{\bb ,0} $ is an $ \eta H_\rm{int} $-invariant analytic vector space, and the essentially self-adjointness of $ \eta H_\rm{int} $ follows.
\end{proof}

We denote the closure of $ H_\rm{int} $ by the same symbol. 

\begin{Prop} Let $ \hat{\chi _\rm{ph}} \in D(\omega ^{-1}) $. Then, $ H_\rm{QED} $ is $ \eta $-self-adjoint.
\end{Prop}

\begin{proof} By the Kato-Rellich theorem and Lemma \ref{eta-lem}, it is sufficient to show that the symmetric operator $ \eta H_\rm{int} $ is $ \eta H_\rm{fr} $- bounded with a relative bound less than 1. 

By Lemma \ref{a-est}, we have for $ \Psi \in D(I \otimes H_\rm{ph} ^{1/2}) $, 
\begin{align}
\| I \otimes A_\mu (\xx ) \Psi \| \le \sqrt{2} \| \hat{\chi _\rm{ph}} / \omega \| _{L^2 (\R ^3 _\kk )} \, \| I \otimes H_\rm{ph} ^{1/2} \Psi \| + \frac{1}{\sqrt{2}}\| \hat{\chi _\rm{ph}} / \sqrt{\omega } \| _{L^2 (\R ^3_\kk )} \, \| \Psi \| .
\end{align}
Thus, for all $ \Psi \in \F _\rm{el} \hat{\otimes } D(H_\rm{ph} ^{1/2}) $, one finds
\begin{align}\label{jA2}
\| j^\mu (\xx ) \otimes A_\mu (\xx ) \Psi \| & = \| ( j^\mu (\xx ) \otimes I)( I \otimes A_\mu (\xx ) ) \Psi \| \no \\
& \le M_\rm{el} \left( \sqrt{2} \| \hat{\chi _\rm{ph}} / \omega \| \, \| I \otimes H_\rm{ph} ^{1/2} \Psi \| +
\frac{1}{\sqrt{2}} \| \hat{\chi _\rm{ph}} / \sqrt{\omega } \| \, \| \Psi \|  \right) .
\end{align}
By \eqref{jA2}, we find that for all $ \Psi \in \F _\rm{el} \hat{\otimes } \big( D(H_\rm{ph} ^{1/2}) \cap D(N_\bb ^{1/2})\big)  \; \big( \subset D(H_\rm{int})\big) $,
\begin{align}
\|  H_\rm{int} \Psi \| \le |e| \,\| \chi _\rm{sp} \| _{L^1 (\R ^3 _\xx )} \, M_\rm{el} \left( \sqrt{2} \| \hat{\chi _\rm{ph}} / \omega \| \, \| I \otimes H_\rm{ph} ^{1/2} \Psi \| +\frac{1}{\sqrt{2}} \| \hat{\chi _\rm{ph}} / \sqrt{\omega } \| \, \| \Psi \| \right) .
\end{align}
Since $ \F _\rm{el} \hat{\otimes } \big( D(H_\rm{ph} ^{1/2}) \cap D(N_\bb ^{1/2})\big)  $ is a core of $ I \otimes H_\rm{ph} ^{1/2} $, we conclude that $ H_\rm{int} $ is $ I \otimes H_\rm{ph} ^{1/2}  $- bounded.
As is easily proved, one finds $ I \otimes H_\rm{ph} ^{1/2} $ is infinitesimally small with respect to $ H_\rm{fr} $. 
Therefore, it follows that $ \eta H_\rm{int} $ is infinitesimally small with respect to $ \eta H_\rm{fr} $, hence we get the desired result.
\end{proof}

\subsection{Existence of dynamics}

In order to prove that there exists dynamics of these quantum fields, we have only to see that Theorem \ref{scheq-and-heieq} can be applied to our case by checking that Assumptions \ref{ass1}, \ref{ass2} and \ref{ass4} are valid. We see in what follows that this is indeed the case with $ A = I \ot N_\bb $, $ H_0 = H_\rm{fr} $, $ H_1 = H_\rm{int} $, $ D= \F _{\bb ,0} $.

Hereafter, we omit trivial tensor product
like $ \ot I $ or $ I \ot $ when no confusion may occur and just write, for instance, $ H_\rm{el} $ instead of $ H_\rm{el} \ot I $ and so forth.

\begin{Lem}\label{QEDass1}\begin{enumerate}[(I)]
\item $ N_\bb $ is self-adjoint and non-negative.
\item $ N_\bb $ and $ H_\rm{fr} $ are strongly commuting.
\item The interaction Hamiltonian $ H_\rm{int} $ and its adjoint $ H_\rm{int} ^* $ are relatively bounded with respect to $  N_\bb ^{1/2} $. 
\item For all $ L \ge 0 $ and $ \Psi \in R (E_ {N_\bb } ([0,L])) $, $ H_\rm{int} \Psi $ and $H_\rm{int}^*\Psi $ belong to $ R (E_ {N_\bb } ([0,L+1])) $.
\end{enumerate}
\end{Lem}

\begin{proof} The statements (I) and (II) are well known facts. 

We prove (III). As in the proof of Lemma \ref{int-est}, it can be seen that $ H_\rm{int} $ is $  N_\bb ^{1/2} $- bounded. We prove that $H_\text{int}^*$ is also $H_\bb ^{1/2}$- bounded. Take arbitrary $\Psi\in \F _{\bb ,0}$. By Lemma \ref{int-eta}, we obtain 
\begin{align}
\| H_\text{int} ^* \Psi \| & =  \| H_\text{int} \eta \Psi \| \no \\
& \le M_\rm{int} \| (N_\bb +1) ^{\fr{1}{2}} \eta \Psi \| \no \\
& \le M_\rm{int} \| (N_\bb +1) ^{\fr{1}{2}} \Psi \| .
\end{align}
Since $ \F _{\bb ,0} $ is a core of $ N_\bb^{1/2} $, we have the assertion (III).

Finally, we prove (IV). Notice that the spectrum of the photon number operator $N_\bb$ is a purely discrete set $\{0,1,2,\dots\}$,
and that for all $L\ge 0$ and $[L]$,  the largest integer satisfying $[L]\le L$,
\begin{align}
R(E_{N_\bb}([0,L])) =R(E_{N_\bb}([0,[L]])) =\mathcal{F}_{[L]}.
\end{align}
Suppose $\Psi$ belongs to
$R(E_{N_\bb}([0,L]))$. Then, it is clear that $\Psi$ belongs to $\mathcal{F}_{[L]+1} 
=R(E_{N_\bb}([0,L+1]))$, since the interaction Hamiltonian creates at most one photon.   
By Lemma \ref{int-eta}, and the fact that $\eta$ and $N_\bb$ are strongly commuting and thus $\eta$ preserves
photon number,  it immediately follows that
$H_\rm{int}^*\Psi\in \F_{[L]+1} = R(E_{N_\bb}([0,L+1]))$.
\end{proof}
From Lemma \ref{QEDass1}, we can apply the general theory constructed in the
earlier sections to obtain: 
\begin{Thm}\label{QEDmain1} For each $ \Psi \in \F _{\bb , 0} $, the series
\begin{align}
U(t,t') \Psi := \Psi + (-i) \int _{t'} ^t d\tau _1 \, H_\rm{int}(\tau _1) \Psi + (-i)^2 \int _{t'} ^t d\tau _1 \int _{t'} ^{\tau _1} d\tau _2 \, H_\rm{int} (\tau _1) H_\rm{int} (\tau _2) \Psi + \cdots 
\end{align}
converges absolutely, where each of integrals is strong integral. Furthermore, $ U(t,t') $ has properties stated in Theorems \ref{main-thm1}, \ref{adjoint}, \ref{main-thm2} and \ref{sch-existence}, with $ H_1 $ replaced by $ H_\rm{int} $, $ H_0 $ by $ H_\mathrm{fr} $ and $ D $ by $ \F _{\bb, 0} $. 
\end{Thm}

As we have already seen in the general theory, once $U(t,t')$ is constructed, we 
immediately obtain a time evolution which is generated by 
$W(t):=e^{-itH_\mathrm{fr}} U(t,0)$. Here, recall that we omit the overline to mean the operator closure.
  In the present application to physics,
it should be made sure that this time evolution is physically acceptable, that is,
the probability amplitude is conserved 
\[ \Expect{\Psi | \Phi} = \Expect{W(t)\Psi| W(t)\Phi}, \]
for suitable vectors $\Psi, \Phi$. The following theorem 
is concerned with this aspect:

\begin{Thm}\label{eta-unitary}
The time evolution operator
\begin{align}
W(t):=e^{-itH_\mathrm{fr}}U(t,0)
\end{align}
satisfies
\begin{align}
 W(t)^\dagger \supset W(t)^{-1}. \label{eta-uni1}
\end{align}
In particular, 
\begin{align}
\Expect{W(t)\Psi|W(t)\Phi}=\Expect{\Psi|\Phi} \label{eta-uni2}
\end{align}
for all $\Psi,\Phi \in D(W(t))$.
\end{Thm}

To prove Theorem \ref{eta-unitary}, we prepare some facts.

 \begin{Lem}\label{eta-uni-vec} 
 For each $ \Psi \in \F _{\bb , 0} $,
\begin{align}\label{eta-unitarity}
\eta U(t,t') \eta \Psi = U(t',t) ^* \Psi
\end{align}
holds.
\end{Lem}

\begin{proof} From Lemma \ref{int-eta}, the operator identities $ \eta H_\rm{int} ^* \eta = H_\rm{int} $ and $ \eta H_\rm{fr} \eta =H_\rm{fr} $ hold, which imply
\begin{align}\label{eta-t}
\eta H_\rm{int}(t) ^* \eta = H_\rm{int} (t) . 
\end{align}
We claim that \eqref{eta-t} implies the identity
\begin{align}\label{induction3}
\eta U_n (t,t') \eta \Psi = U_n (t',t) ^* \Psi 
\end{align}
for all $ n =0,1,2,\dots $ and $ \Psi \in \F _{\bb ,0} $. In fact, by Lemma \ref{int-rep-lem}, one finds
\begin{align}
\eta U_n (t,t') \eta \Psi &= (-i)^n \int _{t'} ^t d\tau _1 \dots \int _{t'} ^{\tau _{n-1}} d\tau _n \, H_\rm{int}(\tau _1)^* \dots H_\rm{int} (\tau _n)^* \Psi \no \\
&= (-i)^n \int _{t'} ^t d\tau _n \dots \int _{t'} ^{\tau _{2}} d\tau _1 \, H_\rm{int}(\tau _n)^* \dots H_\rm{int} (\tau _1)^* \Psi \no \\
& = U _n (t',t) ^* \Psi .
\end{align}
Hence \eqref{induction3} is true for all $ n $, and by summing up over $ n=0,1,2,\dots $, we obtain \eqref{eta-unitarity}.
\end{proof}

\begin{Prop}\label{U-inj}
For all $t,t'\in \Real$, $U(t,t')$ is invertible and the operator equality
\begin{align}
 U(t,t')=U(t',t)^{-1}
 \end{align}
 holds.
\end{Prop}
\begin{proof}
Fix $t,t'$.
Firstly, we prove that $U(t,t')$ is injective.
Note that if for all $\Psi\in \mathcal{F}_{\bb,0}$, $\Expect{\Psi|\Phi}=0$, then
it follows that $\Phi=0$. From Lemma \ref{eta-uni-vec}, we have the operator relation
\begin{align}\label{U-rel}
\eta U(t,t') \subset U(t',t)^*\eta,
\end{align}
since $\mathcal{F}_{\bb,0}$ is a core of $U(t,t')$, and since $U(t,t')^*$ is closed.
Now, suppose that $\Phi$ satisfies $U(t,t')\Phi=0$. Let $\Psi\in\mathcal{F}_{\bb,0}$ be arbitrary. It follows that
\begin{align}
0 &= \Expect{U(t,t')\Psi|U(t,t')\Phi} \no\\
&= \Expect{U(t,t')\Psi, \eta U(t,t')\Phi} \no\\
&= \Expect{U(t,t')\Psi, U(t',t)^*\eta \Phi} \no\\
&=\Expect{U(t',t)U(t,t')\Psi | \Phi}\no\\
&=\Expect{\Phi|\Psi},
\end{align}
where we have used \eqref{U-rel} in the third equality, and Theorem \ref{main-thm2} \eqref{associative}
in the last equality. Since $\Psi \in \mathcal{F}_{\bb,0}$ is arbitrarily taken, we find $\Phi=0$. 
This proves that $U(t,t')$ is injective.

Secondly, we prove 
\begin{align}\label{included}
U(t,t')\subset U(t',t)^{-1}.
\end{align}
Take arbitrary $\Psi\in\mathcal{F}_{\bb,0}$. Then, we have from Theorem \ref{main-thm2}, \eqref{associative}
\[ U(t',t)U(t,t')\Psi = \Psi, \]
which implies $\Psi\in D(U(t',t)^{-1})$ and
\[ U(t,t')\Psi = U(t',t)^{-1}\Psi, \q\Psi\in \mathcal{F}_{\bb,0}. \]
But since $\mathcal{F}_{\bb,0}$ is a core of $U(t,t')$, the relation \eqref{included}
follows.

Finally, we prove  
\begin{align}\label{including}
U(t,t')\supset U(t',t)^{-1}.
\end{align}
Let $\Psi\in D(U(t',t)^{-1})=R(U(t',t))$. Then, there is some $\Phi \in D(U(t',t))$
with $\Psi = U(t,t')\Phi $. Since $\mathcal{F}_{\bb,0}$ is a core of $U(t',t)$,
we can choose a sequence $\{\Phi_n\}_n\subset \mathcal{F}_{\bb,0}$
satisfying
\begin{align}
\Phi_n\to\Phi,\q U(t',t)\Phi_n \to U(t',t)\Phi=\Psi,
\end{align}
as $n$ tends to infinity. Therefore, we have
\[ U(t,t') U(t',t)\Phi_n = \Phi_n \to \Phi, \q n\to\infty. \] 
Since $U(t,t')$ is closed, we conclude that $\Psi = U(t',t)\Phi \in D(U(t,t'))$ and
\[ U(t,t')\Psi = U(t,t') U(t',t)\Phi  = \Phi= U(t',t)^{-1} \Psi. \]
This proves \eqref{including}.
\end{proof}

\begin{Prop}\label{W-inverse-prop} For all $t\in\R$, the operator $W(t)$ is injective and 
\begin{align}
W(t)^{-1}=W(-t) \label{W-inverse}
\end{align}
holds as an operator equality.
\end{Prop}

\begin{proof}
Since $e^{-itH_\mathrm{fr}}$ is unitary and $U(t,0)$ is injective by Lemma \ref{U-inj}, one finds $W(t)$ is injective. 

We prove \eqref{W-inverse}. Fix $t\in\R$. For all $\Psi\in D(W(t))=D(U(t,0))$, we have by 
Theorem \ref{main-thm2} \eqref{associative}
\begin{align}
W(t)W(-t)\Psi &= e^{-it H_\mathrm{fr}} U(t,0) e^{itH_\mathrm{fr}} U(-t,0) \Psi \no\\
 &=U(0,-t)U(-t,0) \Psi \no\\
 &=\Psi. 
\end{align} 
This means $\Psi\in D(W(t)^{-1})$ and
\[ W(-t)\Psi = W(t)^{-1}\Psi. \]
Since $\mathcal{F}_{\bb,0}$ is a core of $W(t)$, we conclude $W(-t)\subset W(t)^{-1}$.
The proof of the inverse inclusion is very similar to the proof of \eqref{including}, and
we omit it.
\end{proof}

\begin{proof}[Proof of Theorem \ref{eta-unitary}] Let $\Psi\in\F_{\bb,0}$. Then, we have
\begin{align}
\eta W(t)^* \eta \Psi = W(-t) \Psi.
\end{align}
Thus, by taking closure, it follows that
\begin{align}
\eta W(t)^* \eta  \supset W(-t).
\end{align}
But, we know $W(-t)=W(t)^{-1}$ from Proposition \ref{W-inverse-prop}.
This proves \eqref{eta-uni1}.

Moreover, since $W(t)\Phi\in D(W(t)^{-1})\subset D(\eta W(t)^* \eta)$, we obtain
\begin{align}
\Expect{W(t)\Psi | W(t)\Phi} &= \Expect{W(t)\Psi , \eta W(t)\Phi} \no\\
	&=\Expect{\Psi,  W(t)^* \eta W(t)\Phi} \no\\
	&=\Expect{\Psi| W(t)^{-1}W(t)\Phi}\no\\
	&=\Expect{\Psi | \Phi}.
\end{align}
This completes the proof.
\end{proof}

We next consider the Heisenberg equations of motion for quantum fields $ A_\mu $ and $ \psi _l $.

\begin{Lem}\label{QEDass2} 
\begin{enumerate}[(I)]
\item $A_\mu(0,f)$ and $A_\mu(0,f)^*$ are $N_\bb^{1/2}$- bounded and closed.
\item $\psi_l(0,g)$ and $\psi_l(0,g)^*$ are $N_\bb^{1/2}$- bounded and closed.
\item For all $L\ge 0$, $\Psi\in R(E_{N_\bb}([0,L]))$ implies that $ A_\mu(0,f)\Psi$,
$A_\mu(0,f)^*\Psi$, $\psi_l(0,g)\Psi $, and $\psi_l(0,g)^*\Psi $ belong to $R(E_{N_\bb}([0,L+1]))$.
\end{enumerate}
\end{Lem}

\begin{proof}
The assertion (I) follows from Lemma \ref{a-est}. The closedness is obvious.

The assertion (II) is an immediate consequence of the fact that Dirac fields are bounded operators.

The claim (III) immediately follows from the fact that photon field operators $A_\mu(0,f)$
and $A_\mu(0,f)^*$ create at most one photon and Dirac fields $\psi_l(0,g)$
and $\psi_l(0,g)$ create no photon.
\end{proof}
Put for $t\in \R$ and $\mu=0,1,2,3$, $l=1,2,3,4$,
\begin{align}
A_{0,\mu} (t,f) \Psi &:=e^{itH_\rm{fr}} A_\mu (0,f) e^{-itH_\rm{fr}} \Psi  , \q f \in \S (\R ^3 ) ,  \\
\psi _{0,l} (t,g) \Psi &:= e^{itH_\rm{fr}} \psi _l (0,g) e^{-itH_\rm{fr}}\Psi , \q g \in \S (\R ^3).
\end{align}
\begin{Lem}\label{for-ass5}
\begin{enumerate}[(I)]
\item For each $ \Psi \in D(N_\bb^{1/2}) $ and each $f,g\in\mathcal{S}(\R^3)$, $ A_{0,\mu} (t,f) \Psi $ and $\psi _{0,l} (t,g) \Psi$ are strongly continuously differentiable in $t\in \R$.

\item The strong derivatives of $ A_{0,\mu} (t,f) \Psi $ and $\psi _{0,l} (t,g) \Psi$, which are denoted by 
$ A'_{0,\mu} (t,f) \Psi $ and $\psi' _{0,l} (t,g) \Psi$, are closable for all $t\in \R$.

\item The operators $ A'_{0,\mu} (t,f) \Psi $ and $\psi' _{0,l} (t,g) \Psi$ are $ A^{1/2} $-bounded uniformly in $t\in\R$. That is, there exist 
constants $ c_0, c_1,d_0,d_1 \ge 0 $ such that for all $ t \in \R $ and $ \xi \in D(A^{1/2}) $,
\begin{align}
 \|  A'_{0,\mu} (t,f) \Psi \| &\le c_0 \| N_\bb^{1/2} \Psi \| + c_1 \| \Psi \| ,\\
  \|  \psi' _{0,l} (t,g)\Psi \| &\le d_0 \| N_\bb^{1/2} \Psi \| + d_1 \| \Psi \| .
 \end{align}
\end{enumerate}
\end{Lem}

\begin{proof}
\begin{enumerate}[(I)]
\item It follows from the general theorem \cite[Lemma 4-49]{AraiFock} that
$t\mapsto a_\mu(e^{it\omega}f)\Psi$ and $t\mapsto a_\mu(e^{it\omega}f)^*\Psi$ are strongly differentiable and the strong derivatives become
\begin{align}
\frac{d}{dt}a_\mu(e^{it\omega}f)\Psi&=a_\mu(i\omega e^{it\omega}f)\Psi,\label{a-dif}\\
\frac{d}{dt}a_\mu(e^{it\omega}f)^*\Psi&=a_\mu(i\omega e^{it\omega}f)^*\Psi,\label{a-dif2}
\end{align}
for $\Psi\in D(N_\bb^{1/2})$ and $f\in \mathcal{S}(\R^3)$. From Lemma \ref{a-est} (i) and \eqref{a-dif}, \eqref{a-dif2}, the continuity of the mappings
\[ t\mapsto \frac{d}{dt}a_\mu(e^{it\omega}f)\Psi,\q  t\mapsto \frac{d}{dt}a_\mu(e^{it\omega}f)^*\Psi \]
 are obvious. 
 For the fermion creation and annihilation operators, the operator valued functions
 \begin{align}
t\mapsto B(e^{itE_M}G) ,\q  t\mapsto B(e^{itE_M}G)^* ,\q G\in D(E_M)
 \end{align}
 are continuously differentiable in the operator norm and the derivatives are
 \begin{align}
\frac{d}{dt}B(e^{itE_M}G) &= B(iE_Me^{itE_M}G),\label{psi-dif}\\
\frac{d}{dt}B(e^{itE_M}G)^*&=B(iE_Me^{itE_M}G)^*.\label{psi-dif2}
 \end{align}
  In fact, one finds from the well known estimates for the fermion creation and annihilation operators \eqref{fermi-est} that
 \begin{align}
 \Bigg|\Bigg| \frac{B(e^{i(t+h)E_M}G)-B(e^{itE_M}G)}{h} - B(iE_Me^{itE_M}G)\Bigg|\Bigg|\no
 &\le  \Bigg|\Bigg| \frac{e^{i(t+h)E_M}G-e^{itE_M}G}{h} - iE_Me^{itE_M}G\Bigg|\Bigg|_{\H_\rm{el}}\no\\
& \to 0,
 \end{align}
 as $h\to 0$, and the mappings
 \[ t\mapsto B(iE_Me^{itE_M}G),\q t\mapsto B(iE_Me^{itE_M}G)^* \]
 are continuous in the operator norm. 
 On the other hand, from these facts, the assertion (I) immediately follows.
\item From \eqref{a-dif} and \eqref{a-dif2}, we have
\begin{align}
A'_{0,\mu}(t,f)\Psi = a_\mu\left(\frac{i\omega e^{it\omega} \widehat{f^*}}{\sqrt{2\omega}}\right)\Psi+a^\dagger\left(\frac{i\omega e^{it\omega}\widehat{f}}{\sqrt{2\omega}}\right)\Psi,\q \Psi\in D(N_\bb^{1/2}).
\end{align}
It is clear from this expression that the adjoint operator of $A'_{0,\mu}(t,f)$ is defined on dense subspace $\F_{\bb,0}$ and therefore 
it is closable. From \eqref{psi-dif} and \eqref{psi-dif2}, the closability of $\psi'_{0,l}(t,g)$ is obvious. 
\item This statement clearly follows from Lemma \ref{a-est} (i) and \eqref{fermi-est}.
\end{enumerate}
\end{proof}
Finally, we have arrived at the existence of strong solutions for the Heisenberg equations
of motion for quantum fields, by combining all the results obtained so far. 
We can define from Lemma \ref{QEDass1}, Lemma \ref{QEDass2} and Theorem \ref{scheq-and-heieq} the following filed operators:
\begin{align}
& D(A_\mu (t,f)) := \F _{\bb ,0} , \q A_\mu (t,f) \Psi := W(-t) A_\mu (0,f) W(t) \Psi , \q \Psi \in \F _{\bb ,0} , \q f \in \S (\R ^3 ) ,  \\
& D(\psi _l (t,g)) := \F _{\bb ,0} , \q \psi _l (t,g) \Psi := W(-t) \psi _l (0,g) W(t) \Psi , \q \Psi \in \F _{\bb ,0} , \q g \in \S (\R ^3),
\end{align}
and one concludes from Lemma \ref{for-ass5}: 
\begin{Thm} For all  
$f,g\in\mathcal{S}(\R^3)$, $\mu=0,1,2,3$ and $l=1,2,3,4$, the operator valued functions $ \Real \ni t \mapsto A_\mu (t,f) $ and $ \Real \ni t\mapsto \psi _l (t,g) $ are strong
solutions of the differential equations 
\begin{align}
\frac{d}{dt} A_\mu (t,f) \Psi & = [iH_\rm{QED} , A_\mu (t,f) ] \Psi , \q \Psi \in D(H_\rm{fr}) \cap \F _{\bb ,0}, \\
\frac{d}{dt} \psi _l (t,g) \Psi & = [iH_\rm{QED} , \psi _l (t,g) ] \Psi , \q \Psi \in D(H_\rm{fr}) \cap \F _{\bb ,0} .
\end{align}
\end{Thm}
\noindent
We remark that from Theorem \ref{eta-unitary} the above quantum fields are also written as
\begin{align}
 A_\mu (t,f) \Psi &= W(t)^\dagger A_\mu (0,f) W(t) \Psi , \q \Psi \in \F _{\bb ,0} , \\
 \psi _l (t,g) \Psi &= W(t)^\dagger  \psi _l (0,g) W(t) \Psi , \q \Psi \in \F _{\bb ,0} .
\end{align}

\appendix
\section{A little about the function $ t\mapsto U(t,t')\xi$}
In Section 3, we have constructed the solution of the differential equation
\begin{align}\label{diff-eq}
\frac{d}{dt}\phi(t) = -iH_1(t)\phi(t) ,\q \phi(0)=\xi, 
\end{align}
via the evolution operator $U(t,t')$, 
\[ \phi(t)= U(t,t')\xi .\]
In this appendix, we investigate the properties of $\phi(t)=U(t,t')\xi$ ($\xi\in D$) as
a function of $t\in \Real$. 
Let us denote by $I$ a closed interval in $\Real$, fix $t'\in I$ and set
\[ K:=(A+1)^{1/2} ,\]
for short.
Throughout the appendix, we assume only Assumption \ref{ass1},
which is sufficient to ensure that there is a solution of \eqref{diff-eq}.

\begin{Lem}\label{s-conti}
Suppose that Assumption \ref{ass1} holds. Let $u:I\to D$ be a strongly continuous mapping satisfying 
\begin{align}
u(t)\in V_L,\quad t\in I
\end{align}
for some constant $ L\ge 0 $. Then, the mapping
\begin{align}
I\ni t \mapsto H_1(t)u(t) \in\mathcal{H}
\end{align}
is strongly continuous. Moreover, for all $t\in I$, the strong Riemann integral
\begin{align}
\int_{t'}^{t}ds\, H_1(s)u(s).
\end{align}
belongs to $V_{L+b}$.
\end{Lem}

\begin{proof} Let $s,t\in I$.
\begin{align}
||H_1(t)u(t) - H_1(s)u(s) || &\le ||(H_1(t)-H_1(s))u(t)|| + || H_1(s)(u(t)-u(s)) || \no\\
&\le || (H_1(t)-H_1(s)) K^{-1} \cdot Ku(t) || + ||H_1(s) K^{-1} ||\,||K(u(t)-u(s)) || .\label{H_1(t)u(t)}
\end{align}
Firstly, note that the mapping
\begin{align}
t\mapsto e^{itH_0}H_1K^{-1}e^{-itH_0}
\end{align}
is strongly continuous, since $H_1K^{-1}$ is a bounded operator. Thus, the first term
of \eqref{H_1(t)u(t)} tends to zero as $s\to t$.

Secondly, note also that there is some $L\ge0$ such that
\[ u(t)-u(s)\in V_{L} ,\quad s,t\in I \]
Therefore, the second term of \eqref{H_1(t)u(t)} also vanishes as $s$ tends to $t$.

From Assumption \ref{ass1}, $H_1(t)u(t)\in V_{L+b}$ ($t\in I$). Since $V_{L+b}$ is 
closed, one finds 
\[ \int_{t'}^{t}ds\, H_1(s)u(s)\in V_{L+b}. \]
\end{proof}

From Lemma \ref{s-conti}, it follows that under Assumption \ref{ass1} we may define
a linear transformation $J$ in the linear space
\begin{align}
\mathcal{C}_\infty:=\{u:I\to D\,|\,\text{$u$ is strongly continuous and there is an $L$ such that 
$u(t)\in V_L$ for all $t\in I$} \}
\end{align} 
by 
\begin{align}
(Ju)(t):=\int_{t'}^t ds\,H_1(s)u(s), \q t \in I, 
\end{align}
where $ t' \in I $ is fixed ($ [t' , t] \subset I $ if $ t' < t $; $ [t,t'] \subset I $ if $ t<t' $).

\begin{Lem}\label{basic-esti}
Under assumption \ref{ass1}, for all $n=0,1,2,\dots$ and $u\in\mathcal{C}_\infty$, the estimate
\begin{align}
||(J^n u)(t) || \le \frac{|t-t'|^n}{n!}C^n(L+(n-1)b+1)^{1/2}(L+(n-2)b+1)^{1/2}\dots (L+1)^{1/2}\sup_{s\in I} ||u(s)||
\end{align}
holds, where $L\ge 0$ is a constant depending on $u$ satisfying
\[ u(t)\in V_L ,\quad t\in I.\]

Moreover, 
\begin{align}
(J^nu)(t)\in V_{L+nb},\quad t\in I.
\end{align}
\end{Lem}
\begin{proof}
We prove by induction. When $n=0$,
the assertion is trivial. 

Suppose the lemma is true for some $n$. Then, by noting that
\[ (J^nu)(t) \in V_{L+nb},\quad t\in I, \]
we have
\begin{align}
||(J\circ J^nu)(t)||&\le \int_{t'}^t ds\, ||H_1(s)(J^nu)(s)|| \no\\
&\le\int_{t'}^t ds\, ||H_1K ^{-1} ||\,||K(J^nu)(s)||\no\\
&\le C(L+nb+1)^{1/2} \int_{t'}^t ds\,||(J^nu)(s)||\no\\
&\le C(L+nb+1)^{1/2} \int_{t'}^t ds\,\frac{|s-t'|^n}{n!}C^n(L+(n-1)b+1)^{1/2}\dots (L+1)^{1/2}\sup_{s\in I} ||u(s)||\no\\
&=\frac{|t-t'|^{n+1}}{(n+1)!}C^{n+1}(L+nb+1)^{1/2}(L+(n-1)b+1)^{1/2}\dots (L+1)^{1/2}\sup_{s\in I} ||u(s)||,
\end{align}
and, by Lemma \ref{s-conti},
\begin{align}
(J^{n+1}u)(t)\in V_{L+(n+1)b},\quad t\in I.
\end{align}
Hence, by induction, the assertion follows.
\end{proof}

 Denote the approximated solution by $\phi_n$:
 \[  \phi_n(t)=\sum_{k=0}^n S_n(t,t')\xi=\sum_{k=0}^n \left((-iJ)^n\xi\right)(t), \q \xi \in D, \]
 where in the last expression, $\xi$ is identified with a constant function which belongs to $\mathcal{C}_\infty$.
 As mentioned above, $\phi_n$ belongs to $\mathcal{C}_\infty$ for all $n\in\Natural$. But,
 the limit function $\phi$ does not have to belong to $\mathcal{C}_\infty$. We identify the set of functions to which $\phi$ belongs.

\begin{Def} Let $T$ be a self-adjoint operator. We say that a function $u:I\to D(T)$ is $T$-uniformly integrable if the following two conditions are satisfied:
\begin{enumerate}[(1)]
\item $t\mapsto Tu(t)$ is strongly continuous.
\item $\lim_{\Lambda\to \infty}\sup_{t\in I}||TE_T([\Lambda,\infty))u(t)||=0$.
\end{enumerate}
\end{Def}
 
For $\alpha\ge 0$, we denote
\begin{align}
\mathcal{C}_\alpha:=\{u:I\to D(K^\alpha)\,|\, \text{$u$ is $K^\alpha$- uniformly integrable}\}.
\end{align}
The set $\mathcal{C}_\alpha$ becomes a vector space by naturally defined addition
and scalar multiplication.
For each $u\in\mathcal{C}_\alpha$, we introduce
\begin{align}
|| u ||_{\alpha,\infty} :=\sup_{t\in I} || K^\alpha u(t)||.
\end{align}
This is finite, since $K^\alpha u(\cdot)$ is strongly continuous.
It is clear that $|| \cdot ||_{\alpha,\infty}$ is a norm of $\mathcal{C}_\alpha$. Further,
we have

\begin{Prop}\label{C_alpha-complete}
The normed space $(\mathcal{C}_\alpha, || \cdot ||_{\alpha,\infty})$ is a Banach space.
\end{Prop}

\begin{proof}
We show only the completeness. Take a Cauchy sequence in $(\mathcal{C}_\alpha, || \cdot ||_{\alpha,\infty})$, $\{u_n\}_n$. Then, for all $\epsilon > 0$, there is an $N\in\Natural$ such that
$n\ge m \ge N$ implies that
\begin{align}\label{cauchy-estimate}
|| K^\alpha (u_n(t)-u_m(t)) || <\epsilon,\quad t\in I.
\end{align}
Since $\{K^\alpha u_n(t)\}_n$ is Cauchy in $\mathcal{H}$, there exists some $v(t)\in\mathcal{H}$ 
such that 
\[ \lim_{n\to\infty} K^\alpha u_n(t) = v(t) .\]
On the other hand, since $K^\alpha\ge 1$, one finds for $n,m\ge N$
\begin{align}
||u_n(t)-u_m(t) || &\le || K^{-\alpha} ||\, || K^{\alpha}(u_n(t) - u_m(t)) || \no\\
 &\le || K^{\alpha}(u_n(t) - u_m(t))|| <\epsilon.
\end{align}
Therefore, $\{u_n(t)\}_n$ is also Cauchy and there is some $u(t)\in\mathcal{H}$
such that 
\[ \lim_{n\to\infty}u_n(t) = u(t). \]
Hence, we have $u(t)\in D(K^\alpha)$ and $K^\alpha u(t)=v(t) $ ($t\in I$).
Take $m\to\infty$ in \eqref{cauchy-estimate}. Then, one has for all $n\ge N$,
\begin{align}
|| K^\alpha (u_n(t)-u(t)) || \le \epsilon,\quad t\in I.
\end{align}
This means 
\[ ||u_n-u||_{\alpha,\infty} \le \epsilon \] 
for all $n\ge N$.

We show $u\in \mathcal{C}_\alpha$, that is, $ u $ is $K^\alpha$- uniformly integrable.
Note that $K^\alpha u(\cdot)$ --- a uniform limit of 
$K^\alpha u_n(\cdot)$ --- is strongly continuous. Take an arbitrary $\epsilon>0$.
We can choose $N\in\Natural$ in such a way that $n\ge N$ implies that
\[ \sup_{t\in I} || K^\alpha(u_n(t)-u(t))|| <\epsilon .\]
Then for an arbitrary $\Lambda\ge 0$, one has
\begin{align}
||K^\alpha E_{K^\alpha} ([\Lambda,\infty ))u(t) || &\le ||K^\alpha E_{K^\alpha} ([\Lambda,\infty ))(u_n(t)-u(t)) ||
+||K^\alpha E_{K^\alpha} ([\Lambda,\infty ))u_n(t) || \no\\
&<  \epsilon + ||K^\alpha E_{K^\alpha} ([\Lambda,\infty ))u_n(t) || .
\end{align}
Taking the limit in which $\Lambda\to \infty$, we have
\begin{align}
\sup_{t\in I_T}||K^\alpha E_{K^\alpha} ([\Lambda,\infty ))u(t) || 
&<  \epsilon + \sup_{t\in I_T}||K^\alpha E_{K^\alpha}([\Lambda,\infty ))u_n(t) || \to \epsilon,\quad \Lambda\to \infty.
\end{align}
Since $\epsilon>0$ is arbitrary, one has
\[ \lim_{\Lambda\to\infty}\sup_{t\in I_T}||K^\alpha E_{K^\alpha} ([\Lambda,\infty ))u(t) || =0, \]
namely, $u\in\mathcal{C}_\alpha$.
\end{proof}

The next lemma shows the relation between $\mathcal{C}_\alpha$ and $\mathcal{C}_\infty$.
\begin{Lem}
For all $\alpha\ge0$, $\mathcal{C}_\infty\subset\mathcal{C}_\alpha$ as normed spaces.
Moreover, $\mathcal{C}_\infty$ is dense in $\mathcal{C}_\alpha$ with respect to the 
$||\cdot || _{\alpha,\infty}$ norm topology.
\end{Lem}
\begin{proof}
First, we prove that $\mathcal{C}_\infty\subset\mathcal{C}_\alpha$.
Let $u\in \mathcal{C}_\infty$.
To show that $u$ is $K^\alpha$-uniformly integrable, note that
there is an $L\ge 0$ such that
\[ u(t)\in V_L,\quad t\in I ,\]
by definition of $\mathcal{C}_\infty$.
One has
\[ || K^\alpha (u(t)-u(s)) || \le (L+1)^{\alpha/2} || u(t)-u(s) ||, \]
which shows that $K^\alpha u(\cdot)$ is strongly continuous.
Let $\Lambda > (L+1)^{\alpha/2}$. Then,
\begin{align}
|| K^\alpha  E_{K^\alpha}([\Lambda,\infty)) u(t) ||&=|| K^\alpha  E_A([\Lambda^{2/\alpha}-1,\infty)) u(t) ||=0.
\end{align}
Hence, $u$ is $K^\alpha$-uniformly integrable, that is, $u\in \mathcal{C}_\alpha$.

Next, we prove that $\mathcal{C}_\infty$ is dense in $\mathcal{C}_\alpha$. 
Take any $u\in\mathcal{C}_\alpha$.
Define $\{u_n\}_n\subset \mathcal{C}_\infty$ by
\[ u_n(t) := E_A([0,n))u(t) ,\quad n=1,2,\dots . \]
Then it follows that
\begin{align}
|| u_n -u ||_{\alpha,\infty}&= \sup_{t\in I_T} || K^\alpha(E_A([0,n))-1)u (t) || \no\\
&= \sup_{t\in I_T} || K^\alpha E_A([n,\infty))u (t) || \no\\
&= \sup_{t\in I_T} || K^\alpha E_{K^\alpha }([(1+n)^{\alpha/2},\infty))u (t) || \to 0,\quad n\to \infty.
\end{align}
This proves that $\mathcal{C}_\infty$ is dense. 
\end{proof}

This lemma tells us that $\mathcal{C}_\alpha$ is the completion of $\mathcal{C}_\infty$ with 
respect to the norm $|| \cdot ||_{\alpha,\infty}$.

\begin{Lem} Let $0\le\alpha\le \beta$. Then $\mathcal{C}_\beta\subset \mathcal{C}_\alpha$ and the inclusion mapping 
\[ \iota:\mathcal{C}_\beta \to\mathcal{C}_\alpha \]
is continuous.
\end{Lem}

\begin{proof} Let $u\in \mathcal{C}_\beta$. Then, $K^{\alpha-\beta}$ is bounded and thus the mapping
\[ t\mapsto K^{\alpha}u(t) = K^{\alpha-\beta} K^{\beta}u(t) \]
is strongly continuous. Further, we have
\begin{align}
\sup_{t\in I_T}|| K^{\alpha}E_{K^{\alpha}}([\Lambda,\infty))u(t)||
&\le||K^{\alpha-\beta}||\,\sup_{t\in I_T}|| K^{\beta}E_{K^{\alpha}}([\Lambda,\infty))u(t)||\no\\
&=||K^{\alpha-\beta}||\,\sup_{t\in I_T}|| K^{\beta}E_{K^{\beta}}([\Lambda^{\beta /\alpha },\infty))u(t)||
\to 0\,\quad \Lambda\to\infty.
\end{align}
These imply that $u\in\mathcal{C}_\alpha$. Hence $ C_\beta \subset C_\alpha $.

The second assertion immediately follows from the fact that
\[ ||K^\alpha u(t) || \le ||K^{\alpha-\beta}|| \,||K^\beta u(t)||,\quad u\in\mathcal{C}_\beta. \] 
\end{proof}

\begin{Prop}\label{u_n-cauchy}
The approximated solution 
$\{\phi_n\}_n\subset \mathcal{C}_\infty$ is a Cauchy sequence in $||\cdot || _{\alpha , \infty } $
for all $\alpha\ge0$.
\end{Prop}
\begin{proof}
Fix $\alpha\ge0$. Choose $L$ for $\xi\in D$ in such a way that
\[ \xi \in V_L \]
holds.
Since
\[ \phi_n=\sum_{k=0}^n (-iJ)^k \xi, \]
and from Lemma \ref{basic-esti}, we obtain $\phi_n(t)\in V_{L+nb}$. Therefore, we have
\begin{align}
||\phi_{n+1}-\phi_n ||_{\alpha,\infty} &= \sup_{t\in I_T} ||K^\alpha(\phi_{n+1}(t)-\phi_n(t)) || \no\\
&\le \sup_{t\in I} (L+nb+1)^{\alpha/2}|| \phi_{n+1}(t)-\phi_n(t) || \no\\
&\le 
\frac{|I|^{n+1}}{(n+1)!}C^n (L+nb+1)^{\alpha/2}(L+nb+1)^{1/2}\dots (L+1)^{1/2}||\xi||,
\end{align}
where $ |I| $ denotes the Lebesgue measure of the interval $I$.
By d'Alembert's ratio test, one obtains
\[ \sum_{n=0}^\infty ||\phi_{n+1} - \phi_n ||_{\alpha,\infty} <\infty. \]
Hence, for any $\epsilon>0$,
there is an $N\in \Natural $, such that $n\ge m\ge N$ implies
\[ ||\phi_n-\phi_m||_{\alpha,\infty} \le \sum_{k=m}^{n-1}||\phi_{k+1}-\phi_k||_{\alpha,\infty}<\epsilon .\]
\end{proof}

\begin{Thm}\label{App1}
For any $\xi\in D$, the solution $\phi$, which is the limit of $\phi_n=S_n(\cdot,t')\xi$,  belongs to $\bigcap_{\alpha \ge 0}\mathcal{C}_\alpha$. That is, $U(\cdot,t')\xi$ is $K^\alpha$-uniformly integrable for all $\alpha\ge 0$.
\end{Thm}
\begin{proof}
Fix $\alpha\ge0$. By Proposition \ref{C_alpha-complete} and Proposition \ref{u_n-cauchy},
there is an element $u_\alpha \in C_\alpha $ satisfying
\[ ||\phi_n-\phi_\alpha||_{\alpha,\infty} \to 0. \]
Take $\beta>0$. Without loss of generality, we may assume $\alpha\le\beta$.
Then, we obtain
\begin{align}
||\phi_\alpha-\phi_\beta||_{\alpha,\infty} &\le  ||\phi_\alpha-\phi_n || _{\alpha , \infty } + || \phi_n-\phi_\beta ||_{\alpha,\infty} \no\\
&\le  ||\phi_\alpha-\phi_n ||_{\alpha,\infty} + ||K^{\alpha-\beta}||\, || \phi_n-\phi_\beta ||_{\beta ,\infty} \to 0,\quad n\to\infty,
\end{align}
which shows that $\phi_\alpha=\phi_\beta$. If $\alpha=0$, then $||\phi_n-\phi||_{0 , \infty} \to 0$. Therefore, 
we conclude that for all $\alpha\ge0$, $\phi=\phi_\alpha\in\mathcal{C}_\alpha$.
\end{proof}

\section*{Acknowledgement}
The authors would like to thank Professor Asao Arai for fruitful comments and discussions and for a critical reading of this manuscript. They also thank Dr. Kazuyuki Wada of Hokkaido University for valuable comments.

\bibliography{ref}

\begin{thebibliography}{10}

\bibitem{AraiFock}
A.~Arai.
\newblock {\em Fock spaces and Quantum Fields I, II (in Japanese)}.
\newblock Nippon-Hyoronsha, Tokyo, 2000.

\bibitem{MR2362899}
A.~Arai.
\newblock Heisenberg operators, invariant domains and {H}eisenberg equations of
  motion.
\newblock {\em Rev. Math. Phys.}, 19(10):1045--1069, 2007.

\bibitem{MR2108957}
G.~Da~Prato, P.~C. Kunstmann, I.~Lasiecka, A.~Lunardi, R.~Schnaubelt, and
  L.~Weis.
\newblock {\em Functional analytic methods for evolution equations}, volume
  1855 of {\em Lecture Notes in Mathematics}.
\newblock Springer-Verlag, Berlin, 2004.
\newblock Edited by M. Iannelli, R. Nagel and S. Piazzera.

\bibitem{MR0492656}
J.~D. Dollard and C.~N. Friedman.
\newblock On strong product integration.
\newblock {\em J. Funct. Anal.}, 28(3):309--354, 1978.

\bibitem{MR552941}
J.~D. Dollard and C.~N. Friedman.
\newblock {\em Product integration with applications to differential
  equations}, volume~10 of {\em Encyclopedia of Mathematics and its
  Applications}.
\newblock Addison-Wesley Publishing Co., Reading, Mass., 1979.
\newblock With a foreword by Felix E. Browder, With an appendix by P. R.
  Masani.

\bibitem{MR1721989}
K.-J. Engel and R.~Nagel.
\newblock {\em One-parameter semigroups for linear evolution equations}, volume
  194 of {\em Graduate Texts in Mathematics}.
\newblock Springer-Verlag, New York, 2000.
\newblock With contributions by S. Brendle, M. Campiti, T. Hahn, G. Metafune,
  G. Nickel, D. Pallara, C. Perazzoli, A. Rhandi, S. Romanelli and R.
  Schnaubelt.

\bibitem{MR2533876}
F.~Hiroshima and A.~Suzuki.
\newblock Physical state for nonrelativistic quantum electrodynamics.
\newblock {\em Ann. Henri Poincar\'e}, 10(5):913--953, 2009.

\bibitem{MR683026}
R.~J. Hughes.
\newblock Singular perturbations in the interaction representation. {II}.
\newblock {\em J. Funct. Anal.}, 49(3):293--314, 1982.

\bibitem{MR583242}
R.~J. Hughes and I.~E. Segal.
\newblock Singular perturbations in the interaction representation.
\newblock {\em J. Funct. Anal.}, 38(1):71--98, 1980.

\bibitem{MR1402248}
M.~E. Peskin and D.~V. Schroeder.
\newblock {\em An introduction to quantum field theory}.
\newblock Addison-Wesley Publishing Company Advanced Book Program, Reading, MA,
  1995.
\newblock Edited and with a foreword by David Pines.

\bibitem{MR0493420}
M.~Reed and B.~Simon.
\newblock {\em Methods of modern mathematical physics. {II}. {F}ourier
  analysis, self-adjointness}.
\newblock Academic Press [Harcourt Brace Jovanovich Publishers], New York,
  1975.

\bibitem{MR2097788}
H.~Spohn.
\newblock {\em Dynamics of charged particles and their radiation field}.
\newblock Cambridge University Press, Cambridge, 2004.

\bibitem{MR2412280}
A.~Suzuki.
\newblock Physical subspace in a model of the quantized electromagnetic field
  coupled to an external field with an indefinite metric.
\newblock {\em J. Math. Phys.}, 49(4):042301, 24, 2008.

\bibitem{MR2541206}
T.~Takaesu.
\newblock On the spectral analysis of quantum electrodynamics with spatial
  cutoffs. {I}.
\newblock {\em J. Math. Phys.}, 50(6):062302, 28, 2009.

\bibitem{MR2148466}
S.~Weinberg.
\newblock {\em The quantum theory of fields. {V}ol. {I}}.
\newblock Cambridge University Press, Cambridge, 2005.
\newblock Foundations.

\bibitem{MR2148467}
S.~Weinberg.
\newblock {\em The quantum theory of fields. {V}ol. {II}}.
\newblock Cambridge University Press, Cambridge, 2005.
\newblock Modern applications.

\end{thebibliography}

\end{document}